\documentclass{article}

\usepackage{amsmath,amsfonts,amssymb,bbm,amsthm} 
\usepackage{mathtools}
\usepackage{fullpage}

\usepackage{microtype}
\usepackage{graphicx}
\usepackage{subfig}
\usepackage{booktabs} 
\usepackage{enumitem}
\usepackage{authblk}
\usepackage{natbib}
\DeclareMathOperator{\Tr}{Tr}

\usepackage{hyperref}

\usepackage{Definitions}

\usepackage{algorithm}
\usepackage{algorithmic}

\newcommand{\xhdr}[1]{\noindent{{\bf #1.}}}

\graphicspath{{./FIG/}}

\begin{document}

\title{Non-submodular Function Maximization \\ subject to a Matroid Constraint, with Applications}

\author{Kashayar Gatmiry$^{1}$}
\author{Manuel Gomez Rodriguez$^{2}$}
\affil{$^{1}$Sharif University, kgatmiry@ce.sharif.edu \\ $^{2}$Max Planck Institute for Software Systems, manuelgr@mpi-sws.org}

\date{}

\clubpenalty=10000
\widowpenalty = 10000

\maketitle

\begin{abstract}
The standard greedy algorithm has been recently shown to enjoy approximation guarantees for constrained non-submodular 
nondecreasing set function maximization. While these recent results allow to better characterize the empirical success of the 
greedy algorithm, they are only applicable to simple cardinality constraints.
In this paper, we study the problem of maximizing a non-submodular nondecreasing set function subject to a general matroid 
constraint.
We first show that the standard greedy algorithm offers an approximation factor of $\frac{0.4 {\gamma}^{2}}{\sqrt{\gamma r} + 1}$, 
where $\gamma$ is the sub\-mo\-du\-la\-ri\-ty ratio of the function and $r$ is the rank of the matroid.
Then, we show that the same greedy algorithm offers a constant approximation factor of $(1 + 1/(1-\alpha))^{-1}$, where $\alpha$ is the 
generalized curvature of the function. 
In addition, we demonstrate that these approximation guarantees are applicable to several real-world applications in which the submodularity
ratio and the generalized curvature can be bounded. 
Finally, we show that our greedy algorithm does achieve a competitive performance in practice using a variety of experiments on synthetic and 
real-world data.
\end{abstract}

\section{Introduction}
\label{sec:introduction}
The problem of maximizing a nondecreasing set function emerges in a wide variety of important real-world applications such as feature selection, sparse modeling, experimental design, 
graph inference and link recommendation, to name a few.
If the set function of interest is nondecreasing and satisfies a natural diminishing property called submodularity\footnote{\scriptsize A set function $F(\cdot)$ is submodular iff it satisfies
that $F(\Acal \cup \{v\}) - F(\Acal) \geq F(\Bcal \cup \{v\}) - F(\Bcal)$ for all $\Acal \subseteq \Bcal \subset \Vcal$ and $v \in \Vcal$, where $\Vcal$ is the ground set.}, 
the problem is (relatively) well understood.
For example, under a simple cardinality constraint, it is known that the \emph{standard} greedy algorithm enjoys an approximation factor of $(1-1/e)$~\citep{nemhauser1978analysis, vondrak2008optimal}. 
Moreover, this constant factor has been improved using the curvature~\citep{conforti1984submodular, vondrak2010submodularity} of a submodular function, which 
quantifies how close is a submodular function to being modular.
Under a general matroid constraint, a variation of the standard greedy algorithm yields a $1/2$-approximation~\citep{fisher1978analysis} and, more 
recently, it has been shown that there exist polynomial time algorithms that yield a $(1-1/e)$-approximation~\citep{calinescu2011maximizing, filmus2012tight}.

However, there are many important applications, from subset selection~\citep{altschuler2016greedy}, sparse recovery~\citep{candes2006stable} and dictionary 
selection~\citep{das2011submodular} to experimental design~\citep{krause2008near}, where the corresponding set function is not submodular.
In this context,~\citet{bian2017guarantees} have shown that, under a cardinality constraint, the standard greedy algorithm enjoys an 
approximation factor of of $\frac{1}{\alpha}(1-e^{-\gamma \alpha})$, where $\gamma$ is the submodularity ratio~\citep{das2011submodular} of the set function, 
which characterizes how close is the function to being submodular, and $\alpha$ is the curvature of the set function. Very recently,~\citet{harshaw2019submodular} 
have also shown that there is no polynomial algorithm with better guarantees.
However, the problem of maximizing a non submodular nondecreasing set function subject to a general matroid constraint has only been studied very recently 
by~\citet{chen2017weakly}, who have shown that a randomized version of the standard greedy algorithm enjoys an approximation factor of $(1+1/\gamma)^{-2}$, 
where $\gamma$ is again the submodularity ratio.
In this paper, we make the following contributions: 
\begin{itemize}
\item[(i)] We show that the standard greedy algorithm~\citep{} yields an approximation factor of $\frac{0.4 {\gamma}^{2}}{\sqrt{\gamma r} + 1}$, where $\gamma$ is the 
submodularity ratio and $r$ is the rank of the matroid. 
While this approximation factor is worse than the one by~\citet{chen2017weakly}, this result shows that the standard greedy algorithm, which is deterministic 
and simpler, does enjoy non trivial theoretical guarantees.
\item[(i)] We show that the standard greedy algorithm yields a constant approximation factor of $(1 + 1/(1-\alpha))^{-1}$, where $\alpha$ is the generalized 
curvature of the function, as defined in previous work~\citep{lehmann2006combinatorial, hassani2017gradient,bogunovic2018robust}.
\item[(ii)] We show that the approximation guarantees from our theoretical analysis is applicable in a wide range of real-world applications, including tree structured 
Gaussian graphical model estimation and visibility maximization in link recommendation, in which the generalized curvature can be bounded. 
\item[(iii)] We show that the standard greedy algorithm does achieve a competitive performance in practice using a variety of experiments on synthetic and 
real-world data.
\end{itemize}
Here, we focus on $\gamma$-weakly submodular functions~\citep{bian2017guarantees, das2011submodular}, however, we would like to acknowledge that there are 
other types of non-submodular set functions that have been studied in the literature in recent years, namely,
approximately submodular functions~\citep{krause2008near},
weak submodular functions~\citep{borodin2014weakly},
set functions with restricted and shifted submodularity~\citep{du2008analysis},  
and
$\varepsilon$-approximately submodular functions~\citep{horel2016maximization}.
Moreover, it would be interesting to extend our study to robust non-submodular function maximization~\citep{bogunovic2018robust}.

\xhdr{Notation} 
We use capital italic letters to denote sets and we refer to $\Vcal$ as the ground set. We use $F(\cdot)$ to represent a set function and define the marginal gain function 
$\rho_{\Omega}$ of each subset $\Omega \subseteq \Vcal$ as $\rho_{\Omega}(\Scal) = F(\Scal \cup \Omega) - F(\Scal), \forall \Scal \subseteq \Vcal$.
Whenever $\Omega = \{v\}$ is a singleton, we use the symbol $\rho_{v}$ instead of $\rho_{\{v\}}$ for simplicity.

\section{Preliminaries}
\label{sec:preliminaries}
In this section, we start by revisiting the definitions of matroids and $\gamma$-weakly submodular functions~\citep{bian2017guarantees,das2011submodular}. 
Then, we define $\alpha$-submodular functions, a subclass of $\gamma$-weakly submodular functions defined in terms of the generalized curvature 
$\alpha$~\citep{lehmann2006combinatorial, hassani2017gradient,bogunovic2018robust}. 
Finally, we establish a relationship between $\alpha$-submodular functions and set functions representable as a difference between submodular functions.

Matroids are combinatorial structures that generalize the notion of linear independence in matrices. More formally, a matroid can be defined as 
follows~\citep{fujishige2005submodular, schrijver2003combinatorial}:
\begin{definition}
A matroid $\Mcal = (\Vcal, \Ical)$ is a pair defined over the ground set $\Vcal$ and a family of sets (the independent sets) $\Ical$ that satisfies three axioms:
\begin{enumerate}[leftmargin=0.5cm,noitemsep,nolistsep]
\item Non-emptiness: the empty set $\emptyset \in \Ical$.
\item Heredity: if $\Ycal \in \Ical$ and $\Xcal \subseteq \Ycal$, then $\Xcal \in \Ical$.
\item Exchange: if $\Xcal \in \Ical$, $\Ycal \in \Ical$ and $|\Ycal| > |\Xcal|$, then there exists $z \in \Ycal \backslash \Xcal$ such that $\Xcal \cup \{ z \} \in \Ical$.
\end{enumerate}
\end{definition}
The rank of a matroid is the maximum size of an independent set in the matroid. 

A $\gamma$-weakly submodular function is defined in terms of the submodularity ratio $\gamma$:
\begin{definition}
A set function $F(\cdot)$ is $\gamma$-weakly submodular if
\begin{equation} \label{eq:weak-submodularity}
\sum_{v \in \Omega \backslash \Scal} \rho_{v}(\Scal) \geq \gamma \rho_{\Omega} (\Scal), \quad \forall \Omega, \Scal \subseteq \Vcal, 
\end{equation} 
where the largest $\gamma \leq 1$ such that the above inequality is true is called submodularity ratio. Submodular functions have submodularity
ratio $\gamma = 1$.
\end{definition}
A $\alpha$-submodular function is defined in terms of the generalized curvature $\alpha$:
\begin{definition}
A set function $F$ is $\alpha$-submodular if, for any $v \in \Vcal$ and subsets $\Acal \subseteq \Bcal \subseteq \Vcal$,
  \begin{align}  \label{eq:xi-submodularity}
  \rho_{v}(\Acal) \geq (1-\alpha) \rho_{v}(\Bcal), 
  \end{align}
where the smallest $\alpha \leq 1$ such that the above inequality is true is called the generalized curvature.
\end{definition}
As shown very recently~\citep{bogunovic2018robust,halabi2018combinatorial}, there is a relationship between $\alpha$-submodular functions 
and $\gamma$-weakly submodular functions:
\begin{proposition}
Given a set function $F$ with generalized curvature $\alpha$, then it has submodularity ratio $\gamma \geq 1-\alpha$.
\end{proposition}
Moreover, the following proposition establishes a relationship between $\alpha$-submodular functions and (nondecreasing) set functions representable as a difference 
between submodular functions\footnote{\scriptsize Note that any set function can be expressed as a difference between two submodular functions~\citep{narasimhan2012submodular}.}:
\begin{proposition} \label{thm:diff-submodular}
Given a set function $G = F_1 - F_2$, where $F_1$ and $F_2$ are nondecreasing submodular functions and let $0 < \alpha \leq 1$ be the smallest
constant\footnote{\scriptsize Note that, if $G$ is nondecreasing, such constant $1-\alpha$ will always exist.} such that 
\begin{equation} \label{eq:ds}
F_2(\Scal \cup \{v\}) - F_2(\Scal) \leq \alpha [F_1(\Scal \cup \{v\}) - F_1(\Scal)]
\end{equation}
for all $\Scal \subseteq \Vcal$ and $v \in \Vcal$.
Then, $G$ has generalized curvature $\alpha$.
\end{proposition}
\begin{proof}
Let $\Acal \subseteq \Bcal$. Then, we have:
\begin{align*}
G(\Acal \cup \{v\}) - G(\Acal) &= F_1(\Acal \cup \{v\}) - F_1(\Acal) - F_2(\Acal \cup \{v\}) + F_2(\Acal) \\
&\geq (1-\alpha)[F_1(\Acal \cup \{v\}) - F_1(\Acal)] \geq (1-\alpha)[F_1(\Bcal \cup \{v\}) - F_1(\Bcal)] \\
&\geq (1-\alpha)[F_1(\Bcal \cup \{v\}) - F_1(\Bcal) - F_2(\Bcal \cup \{v\}) + F_2(\Bcal)] = (1-\alpha) [G(\Bcal \cup \{v\}) - G(\Bcal)],
\end{align*}
where the first inequality follows from the definition of $\alpha$ and the second and third inequalities follow from the submodularity of $F_1$ and the
monotonicity of $F_2$, respectively.
\end{proof}
In general, set functions representable as a difference between submodular functions cannot be approximated in polynomial time~\citep{iyer2012algorithms}, 
however, the above result identifies a particular class of these functions for which the standard greedy algorithm achieves approximation guarantees. 

\xhdr{Remarks} The notion of $\alpha$-submodularity is fundamentally different from $\varepsilon$-approxi\-ma\-te submodularity~\citep{horel2016maximization}.
%
%
%
More specifically, there exist $\varepsilon$-approxi\-ma\-te\-ly sub\-modular func\-tions that are not $\alpha$-submodular for any $\alpha$
arbitrarily close to $1$. 
Define the following set functions $G$ and $F$ over $\Vcal = \{ 1,2,3\}$:
\begin{align*}
G(\Scal) = \begin{cases}0 & |\Scal| = 0\\1-\epsilon & |\Scal| = 1\\ 1 & |\Scal| = 2\\ 1 + \epsilon & |\Scal|=3 \end{cases} \quad \quad 
F_{\delta}(\Scal) = \begin{cases}0 & |\Scal| = 0\\1-\epsilon & |\Scal| = 1\\ 1-\epsilon+\delta & |\Scal| = 2\\ 1 + \epsilon & |\Scal|=3 \end{cases}.
\end{align*}
Then, it can be readily shown that for $0 \leq \epsilon \leq \frac{1}{2}$, $G$ is submodular and, if $0 \leq \delta \leq \epsilon$, then $\forall \Scal \subset \Vcal, (1 - \epsilon)G(\Scal) \leq F_{\delta}(\Scal) \leq (1 + \epsilon)G(\Scal)$, which implies that $F_{\delta}$ is $\epsilon$-approxi\-ma\-tely submodular. However,
\begin{align}
\frac{\rho_{2}(\{ 1 \})}{\rho_{2}(\{ 1,3 \})} = \frac{F_{\delta}(\{ 1,2 \}) - F_{\delta}(\{ 1\})}{F_{\delta}(\{ 1,2,3\}) - F_{\delta}(\{ 1,3\})} \leq \frac{\delta}{\epsilon} \rightarrow 0,
\end{align}
as $\delta \rightarrow 0$. Therefore, the generalized curvature of $F_{\delta}$ can approach $1$ arbitrarily, which proves our claim.

\section{Approximation Guarantees}
\label{sec:formulation}
In this section, we show that the standard greedy algorithm~\citep{nemhauser1978analysis} (Algorithm~\ref{alg:greedy}) enjoys approximation guarantees at maximizing set non-submodular 
nondecreasing functions $F$ under a matroid constraint $\Mcal = (\Vcal, \Ical)$ with rank $r$, \ie,
\begin{align} \label{eq:optimization-problem}
\underset{\Scal \in \Ical}{\text{maximize}} & \quad F(\Scal) 
\end{align}
More specifically, we first show that the greedy algorithm offers an approximation factor of $\frac{0.4 {\gamma}^{2}}{\sqrt{\gamma r} + 1}$ for $\gamma$-weakly submodular 
functions whenever $r \geq 3$. 
Here, note that, whenever $r \leq 2$, we can check all of the sets in our matroid using brute force without increasing the time complexity of the greedy
algorithm and, hence, we omit those cases from our analysis.
Then, we show that the greedy algorithm enjoys an approximation factor of $(1 + 1/(1-\alpha))^{-1}$ for $\alpha$-submodular functions independently of the rank of 
the matroid. 
\begin{algorithm}[t] \label{alg:greedy}
\renewcommand{\algorithmicrequire}{\textbf{Input:}}
\renewcommand{\algorithmicensure}{\textbf{Output:}}
\begin{algorithmic}[1]
\caption{Greedy algorithm}\label{alg:greedy}
\REQUIRE Ground set $\Vcal$, matroid $\Mcal = (\Vcal, \Ical)$, non-submodular nondecreasing set function $F$ 
\ENSURE Set of items $\Scal_n$
\STATE $\Scal_{0} \leftarrow \varnothing, \Ucal_{0} \leftarrow \varnothing$
\STATE $t \leftarrow 1$
\WHILE{$|\Ucal_{t-1}| < |\Vcal|$}
\STATE $\Ucal_i \leftarrow \Ucal_{t-1}$
\REPEAT
	\STATE $v^* \leftarrow \text{argmax}_{v \in \Vcal \backslash \Ucal_{t}} \rho_{v}(\Scal_{t-1})$
	\STATE $\Ucal_{t} \leftarrow \Ucal_{t} \cup \{ v^* \}$ \mbox{\quad \,\,\,\, \% Item is considered}
\UNTIL{$S_{t-1} \cup \{ v^* \} \in \Ical$}
\STATE $S_{t} \leftarrow S_{t-1} \cup \{ v^* \}$ \mbox{\quad \% Item is selected}
\STATE $t \leftarrow t+1$
\ENDWHILE \\
\STATE \mbox{\bf return} $S_{t-1}$
\end{algorithmic}
\end{algorithm}

%

\subsection{$\gamma$-weakly submodular functions}
Our main result is the following theorem, which shows that the greedy algorithm achieves an approximation factor 
that depends on the rank of matroid:
\begin{theorem} \label{thm:guarantee-weakly-submodular}
Given a ground set $\Vcal$, a matroid $\Mcal = (\Vcal, \Ical)$ with rank $r \geq 3$ and a non-decreasing $\gamma$-weakly submodular set function 
$F$. Then, the greedy algorithm, summarized in Algorithm~\ref{alg:greedy}, returns a set $\Scal$ such that
\begin{equation}
F(\Scal) \geq \frac{0.4 {\gamma}^{2}}{\sqrt{\gamma r} + 1} OPT
\end{equation}
where $OPT$ is the optimal value.
\begin{proof}
Let $\Scal_t$ be the set of items selected by the greedy algorithm in the first $t$ steps, assume $\Scal_0 = \emptyset$, and define 
$K_t = OPT - F(\Scal_t)$.
The core of our proof lies on the following key Lemma (proven in Appendix~\ref{app:lem-key-weakly-submodular}), which shows that, if $K_t$ is large enough, then $K_{t+1}$ will 
be smaller than a factor of $K_t$:
\begin{lemma} \label{lem:key-weakly-submodular}
Suppose that $\frac{K_t}{OPT} \geq \alpha^{*}$, where $\alpha^* = \frac{1}{\frac{\gamma^2}{2(\sqrt{\gamma r}+1)} + 1}$. Then, it holds that
\begin{equation} 
K_{t+1} \leq (1-\theta) K_{t}, \label{eq:recursive-definition}
\end{equation}
where 
\begin{equation} \label{eq:thetadefinition}
\theta = \sqrt{\frac{1}{\gamma} \frac{\log(\frac{1}{\alpha^*}) - 1 + \alpha^*}{r \alpha^*}},
\end{equation}
\end{lemma}
Given the above Lemma, we proceed as follows. 
First, we note that the function $\frac{\gamma^2}{2(\sqrt{\gamma r}+1)}$ is increasing with respect to $\gamma$, it is decreasing with respect to $r \geq 1$ and 
$\gamma \leq 1$ by definition. Therefore, it follows that
\begin{equation} \label{eq:alphalowerbound}
\alpha^* = \frac{1}{\frac{\gamma^2}{2(\sqrt{\gamma r}+1)} + 1} \geq \frac{1}{\frac{1}{2(\sqrt{r} +1)}+ 1} \geq \frac{1}{\frac{1}{2(1 + 1)}+ 1} = 0.8. 
\end{equation}
and thus
\begin{equation} \label{eq:alphaupperbound}
1 - \alpha^* = \frac{\frac{\gamma^2}{2(\sqrt{\gamma r}+1)}}{\frac{\gamma^2}{2(\sqrt{\gamma r}+1)} + 1}  \geq 0.4 \frac{\gamma^2}{\sqrt{\gamma r}+1} 
\end{equation}
Second, we note that the function $g(\alpha) = \frac{\alpha^*(\ln(\frac{1}{\alpha^*}) - 1 + \alpha^*)}{(1-\alpha^*)^2}$ is decreasing. The reason is that $g^{\prime}(\alpha) = \frac{\log(\frac{1}{\alpha})(1+\alpha) - 2(1-\alpha)}{(1-\alpha)^3} \geq 0$, because for $\alpha \leq 1$, we have $\log(\frac{1}{\alpha}) \geq \frac{2(1-\alpha)}{1+\alpha}$. Therefore, using $r \geq 3$, for $\omega := \sqrt{\frac{r}{\gamma}} \sqrt{\frac{\alpha^*(\ln(\frac{1}{\alpha^*}) - 1 + \alpha^*)}{(1-\alpha^*)^2}}$,
\begin{equation} \label{eq:omegaequation}
w \geq \sqrt{3 g(\alpha^*)} \geq \sqrt{3 g(0.8)} \geq \sqrt{3 \times 0.46} \geq 1,
\end{equation}
where we used the lower bound of $\alpha^*$ from Eq.~\ref{eq:alphalowerbound}. Next, using that $(1-\theta)^{\frac{1}{\theta}} \leq e^{-1}$ and $e^{-y} \leq \frac{1}{y+1}$ for all $y \geq 0$, we have that
\begin{equation} \label{eq:alphaomega}
(1-\theta)^r = [(1-\theta)^{\frac{1}{\theta}}]^{\theta r} \leq e^{-\theta r} \leq \frac{1}{\theta r + 1} = \frac{\alpha^*}{\alpha^* \theta r + \alpha^*} \leq \frac{\alpha^*}{\omega (1-\alpha^*)+\alpha^*} \leq 1,
\end{equation}
where the last inequality follows from Eq.~\ref{eq:omegaequation}. 

Now, if we combine Eq.~\ref{eq:omegaequation} and Eq.~\ref{eq:alphaomega}, we conclude that $(1-\theta)^{r} \leq \alpha^{*}$.
Hence, if we could use Eq.~\ref{eq:recursive-definition} up until step $r$ (the last step\footnote{Since $F$ is monotone nondecreasing, $\Scal$ has cardinality 
equal to the dimension of the matroid, \ie, $|\Scal| = r$}), then we would conclude that $\frac{K_r}{OPT} \leq \alpha^{*}$. 
However, since we can only use Eq.~\ref{eq:recursive-definition} whenever $\frac{K_t}{OPT} \geq \alpha^{*}$, we conclude that there exists some $t^{*}$ such that 
$\frac{K_{t^{*}-1}}{OPT} \geq \alpha^{*}$ and $\frac{K_{t^{*}}}{OPT} \leq \alpha^{*}$. Thus,
\begin{equation*}
F(\Scal_{t^*}) = OPT - K_{t^*} \geq OPT(1 - \alpha^*) \geq  0.4 \frac{\gamma^2}{\sqrt{\gamma r}+1} OPT,
\end{equation*}
where we have also used Eq.~\ref{eq:alphaupperbound}. Finally, since $F$ is monotone nondecreasing, it follows that
\begin{equation*}
F(\Scal) = F(\Scal_{r}) \geq F(\Scal_{t^*}) \geq 0.4 \frac{\gamma^2}{\sqrt{\gamma r}+1}OPT,
\end{equation*}
which concludes the proof.
\end{proof}
\end{theorem}

\begin{corollary} \label{cor:guarantee-weakly-submodular}
The greedy algorithm enjoys an approximation guarantee of $\Omega(\gamma^2)$ if $\gamma r = O(1)$ and $\Omega(\gamma \sqrt{\gamma}/\sqrt{r})$ if 
$\gamma r = \Omega(1)$.
\end{corollary}

\xhdr{Remarks}
We would like to acknowledge that the randomized algorithm recently introduced by~\citet{chen2017weakly} enjoys better approximation guarantees
at maximizing $\gamma$-weakly submodular functions, however, we do think that the above result has some value. More specifically:
\begin{itemize}
\item[(i)] Our theoretical result shows that the greedy algorithm, which is deterministic and simpler, does enjoy non trivial theoretical 
guarantees. These theoretical guarantees supports its strong empirical performance in several applications (\eg, tree-structured Gaussian 
graphical model estimation).
\item[(ii)] If $\gamma r = O(1)$, the approximation factors of the greedy algorithm and the algorithm by~\citet{chen2017weakly} are both $\Omega(\gamma^2)$ 
and of the same order. 
\item[(iii)] The proof technique used in Theorem~\ref{thm:guarantee-weakly-submodular} is novel and it may be useful in proving better approximation 
factors of other randomized algorithms for maximizing non submodular set functions.
\end{itemize}

\subsection{$\alpha$-submodular functions}
Our main result is the following theorem, which shows that the greedy algorithm achieves an approximation factor 
that is independent of the rank of the matroid:
\begin{theorem}  \label{thm:guarantee-xi-submodular}
Given a ground set $\Vcal$, a matroid $\Mcal = (\Vcal, \Ical)$ and a nondecreasing $\alpha$-submodular set function $F$.
%
Then, the greedy algorithm returns a set $\Scal$ such that $F(\Scal) \geq OPT/(1 + 1/(1-\alpha))$, where $OPT$ is the optimal value.
\begin{proof}
Let $\Tcal$ be the optimal set of items and $\Scal$ be the set of items selected by the greedy algorithm.
Moreover, let  
%
$\Ucal_t$ be the items considered by the algorithm in the first $t$ steps,
$\Scal_t = \{s_i\}_{i=1}^{t}$ be the items selected by the greedy algorithm in the first $t$ steps in order of their consideration, 
and $\Tcal_t = \{ t_{i} \}_{i=1}^{q_t}$ be the items in $\Tcal$ considered by the greedy algorithm in the first $t$ steps also in order of their consideration.
%
%

%
%
According to the definition of the greedy algorithm, adding any element from $\Ucal_{t} \backslash \Scal_{t}$ to $\mathcal{S}_{t}$ violates the matroid constraint $\Ical$ (otherwise, 
that element should have been picked by the greedy algorithm). 
Thus, $\Ucal_{t} \subseteq \text{span}(\Scal_{t})$, which implies $\text{rank}(\Ucal_{t}) = \text{rank}(\Scal_{t})$. Moreover, $\Tcal_{t} \subseteq \Ucal_{t}$ implies 
$\text{rank}(\Tcal_{t}) \leq \text{rank}(\Ucal_{t})$, therefore, $\text{rank}(\Tcal_{t}) \leq \text{rank}(\Scal_{t})$.
%
%
However, $\Scal_{t}$ and $\Tcal_{t}$ are both independent sets of the matroid $\Ical$, because they are both feasible solutions. As a result, it follows that $\text{rank}(\Tcal_{t}) = q_t$,
$\text{rank}(\Scal_{t}) = t$, and thus $q_t \leq t$.
%
%
Moreover, this implies that 
$s_{i}$ is considered in the greedy algorithm at some point before $t_{i}$. This means that, at the point that the greedy picks $s_{i}$, $t_{i}$ does not have a higher marginal 
gain than $s_{i}$, \ie, $\rho_{s_{i}}(\Scal_{i-1}) \geq \rho_{t_{i}}(\Scal_{i-1}),\, \forall 1 \leq i \leq n$.
%
%
Hence, we can write
\begin{align*}
F(\Tcal_t)  & \leq  F(\Scal_{t} \cup \Tcal_t) = F(\Scal_{t}) + \sum_{i=1}^{q_t} \rho_{t_{i}}(\Scal_{t} \cup \{t_1, \ldots, t_{i-1}\} ) \leq F(\Scal_{t}) + \sum_{i=1}^{q_t} \frac{1}{1-\alpha} \rho_{t_{i}}(\Scal_{i-1}) \\
& \leq F(\Scal_{t-1}) + \frac{1}{1-\alpha} \sum_{i=1}^{q_t} \rho_{s_{i}}(\Scal_{i-1}) \leq  F(\Scal_{t}) +  \frac{1}{1-\alpha} \sum_{i=1}^{t} \rho_{s_{i}}(\Scal_{i-1}) = \left(1 + \frac{1}{1-\alpha}\right) F(\Scal_t), 
\end{align*}
where, in the first inequality, we have used the monotonicity of $F$ and, in the second inequality, we have used the $\alpha$-submodularity of $F$.
This concludes the proof.
%
\end{proof}
\end{theorem}
\xhdr{Remarks} 
We would like to highlight that the above proof differs significantly from that of Theorem 2.1 in~\citet{nemhauser1978analysis}, which is
significantly more involved.
More specifically, in our proof, we cannot apply proposition 2.2 in~\citet{nemhauser1978analysis} because the decreasing monotonicity 
of the marginal gains of the added elements fails to hold in the absence of the submodularity condition.
As a result, our proof does not resort to linear program duality and it is generalizable for the case of having an intersection of $P$ matroids 
instead of just one.

\section{Applications}
\label{sec:applications}
In this section, we consider several real-world applications and their corresponding $\gamma$-weakly submodular and $\alpha$-submodular functions and
matroid constraints. 
We demonstrate that the submodularity ratio and the generalized curvature can be bounded and, as a result, our approximation guarantees are applicable.

\subsection{$\gamma$-weakly submodular functions} 
\xhdr{Tree-structured Gaussian graphical models}
Gaussian graphical models (GGMs) are widely used in many applications, \eg, gene regulatory networks~\citep{friedman2000using, friedman2004inferring, irrthum2010inferring}. 
A GGM is typically characterized by means of the sparsity pattern of the inverse of its covariance matrix $\Sigma^{-1}$. 
Here, we look at maximum likelihood estimation (MLE) problem for tree-structured GMMs~\citep{tan2010learning} from the perspective of 
$\gamma$-weakly submodular maximization. 
More specifically, given a set of $n$-dimensional samples $\{ \xb_1, \ldots, \xb_{N} \}$ and a lower and upper bound on the eigenvalues $\sigma_i(\Sigma)$ of the true covariance 
matrix $\Sigma$, \ie, $L \leq \sigma_i(\Sigma) \leq U$, we can readily rewrite the MLE problem as: 
\begin{equation} \label{eq:mle-ggm}
\underset{\Tcal \in \Tcal^{n}}{\text{maximize}} \quad F(\Tcal)
\end{equation}
with
\begin{equation*}
F(\Tcal) = \max_{\substack{\Sigma^{-1} \in \mathcal{R}(S^n_{++} , \Tcal) \\ U^{-1} \leq \sigma_i(\Sigma^{-1}) \leq L^{-1}}} N \log | \Sigma^{-1}| - \Tr \left(\Sigma^{-1} \sum_{i=1}^{N} x_i x_i^{T}\right),
\end{equation*}
where $\Tcal^{n}$ is the set of all trees with $n$ vertices, $\mathcal{R}(S^n_{++} , \mathcal{T}^{n})$ is the set of $n \times n$ positive definite matrices 
whose sparsity pattern is based on a tree in $\mathcal{T}^n$ and note that, for a fixed $\Tcal$, the optimization problem that defines $F(\Tcal)$ is convex 
with respect to $\Sigma^{-1}$.
Then, the following Theorem (proven by ~\citet{elenberg2016restricted}) and Proposition (proven in Appendix~\ref{app:prop:keyprop}) characterize the submodularity 
ratio $\gamma$ of $F(\Tcal)$:

\begin{theorem} \label{thm:generalclasstheorem}
Let $f: \mathbb{R}^n \rightarrow \mathbb{R}$ be a concave function with curvature based bounds
\begin{equation*}
f(x) + (y-x)^{T}f^{\prime}(x) - R \left \|  y - x\right \|^2 \leq f(y) \leq f(x) + (y-x)^{T}f^{\prime}(x) - r \left \|  y - x\right \|^2.
\end{equation*}
and define the set function $F$ as
\begin{align*}
\forall \Scal \subseteq \{ 1,2,..,n\}, \  F(\Scal) = \max_{\text{supp}(x) \subseteq \Scal} f(x).
\end{align*}
Then, the submodularity ratio of $F$ is $\gamma \geq \frac{r}{R}$.
\end{theorem}
\begin{proposition} \label{prop:keyprop}
The function $F(\Tcal)$ satisfies Theorem~\ref{thm:generalclasstheorem} for $r = \frac{1}{U^2}$ and $R = \frac{1}{L^2}$ and, as a result, its submodularity 
ratio is $\gamma \geq (\frac{L}{U})^2$.
\end{proposition}

\xhdr{Social welfare allocation}
Social welfare maximization has been studied extensively in the context of combinatorial auctions~\citep{feige2009maximizing, feige2006approximation, mirrokni2008tight, vondrak2008optimal}. 
In a popular variant of this problem, given a set of items $\Scal$ and $n$ players, each of them with a monotone utility function $w_{i}: 2^{\Scal} \rightarrow \mathbb{R}^+$, 
the goal is to partition $\Scal$ into disjoint subsets $\Scal_{1}, \ldots ,\Scal_{n}$ that maximize the social welfare $W(\Scal) = \sum_{i=1}^n w_i (\Scal_i)$. 
Since the partition of the items can be viewed as a matroid constraint, \ie, for any valid partitioning, each item is assigned to exactly one player, the above formulation
reduces to the problem of maximizing a set function subject to a matroid constraint.
In this context, Calinescu et al.~\citep{calinescu2011maximizing} has recently proposed a polynomial time algorithm with approximation guarantees whenever $w_i$ are submodular 
functions. 
Here, since the sum of $\gamma$-weakly submodular functions is also $\gamma$-weakly submodular, our results imply that our greedy algorithm enjoys approximation guarantees
whenever $w_i$ are $\gamma$-weakly submodular. 
This includes natural utility functions such as $w_i(\Scal_i) = \sum_{j \in \Scal_i} F_j(x_{ij})$, where $x_{ij}$ is the \emph{amount} of item $j$ by player $i$ and $F_j$ is a 
strongly concave function that satisfies the curvature based bounds in Theorem~\ref{thm:generalclasstheorem}.

\xhdr{LPs with combinatorial constraints}
In a recent work~\citep{bian2017guarantees}, Bian et al. have shown that linear programs with combinatorial constraints, which appeared
in the context of inventory optimization, can be reduced to maximizing $\gamma$-weakly submodular functions. 
More specifically, define a set function $F(\Scal) = \max_{\text{supp}(x) \subseteq \Scal, x \in \mathcal{P}} d^{T} x$, where $\Pcal$ is a 
polytope and $d$ is a given vector. 
Then, they have shown that $F$ has non-zero submodularity ratio that depends on the polytope $\Pcal$. However, they only impose
simple cardinality constraints on $\Scal$. 
Our results imply that the greedy algorithm also enjoys approximation guarantees under a general matroid constraint $\Scal \in \Ical$.

\subsection{$\alpha$-submodular functions}
\xhdr{Visibility optimization in link recommendation}
In the context of viral marketing, a recent line of work~\citep{karimi2016smart, upadhyay2018deep, zarezade2018steering, zarezade2017redqueen}, 
has developed a variety of algorithms to help users in a social network maximize the visibility of the stories they post. 
More specifically, these algorithms find the best times for these users, the \emph{broadcasters}, to share stories with her followers so that 
they elicit the greatest attention.
Motivated by this line of work and recent calls for fairness of exposure in social networks~\citep{biega2018equity, singh2018fairness}, we consider 
the following visibility optimization problem in the context of link recommendation~\citep{lu2011link}:
given a set of candidate links $\Ecal'{}$ provided by a link recommendation algorithm,
the goal is to find the subset of these links $\Ecal_{\Bcal} \subseteq \Ecal'{}$ that maximize the average visibility that a set of broadcasters $\Bcal$ 
achieve with respect to the (new) followers induced by these links\footnote{\scriptsize The followers the broadcasters would gain if $\Ecal_{\Bcal}$ were added.}, 
under constraints on the maximum number of links per broadcaster.
More specifically, we can show that this problem reduces to maximizing a $\alpha$-submodular function (average visibility) under a partition matroid constraint 
(number of links per broadcaster), where the generalized curvature $\alpha$ can be analytically bounded.
For space constraints, we defer most of the technical details to Appendix~\ref{app:visibility}, which also provide additional motivation for the problem, and 
here we just state the main results.

Formally, let measure the visibility a broadcaster $u \in \Bcal$ achieves with respect to the (new) followers induced by the links $\Ecal_{\Bcal}$ as the average 
number of stories posted by her that lie within the top $K$ positions of those followers'{} feeds over time. Here, for simplicity, each follower'{}s feed ranks stories 
in inverse chronological order, as in previous work~\citep{karimi2016smart, zarezade2018steering, zarezade2017redqueen}.
Moreover, assume that\footnote{\scriptsize These assumptions are natural in most practical scenarios, as argued in Appendix~\ref{app:visibility}.}: 
\begin{itemize}[noitemsep,nolistsep,leftmargin=0.8cm]
\item[(i)] the intensities (or rate) $\lambda(t)$ at which broadcasters posts stories and followers receive stories are $\xi$-bounded, \ie, $\sup(\lambda) \leq \xi \inf(\lambda)$; and, 
\item[(ii)] at each time $t$, the intensity at which each broadcaster posts is lower than a fraction $\frac{1}{\rho \sqrt{K-1}}$ of each of her followers'{} feeds intensity, where $\rho$ is a given
constant. 
\end{itemize}
Then, we can characterize the generalized curvature $\alpha$ of the average visibility $F(\Ecal_{\Bcal})$ these broadcasters achieve using the following Proposition:
\begin{proposition} \label{prop:visibility-strong-submodularity-ratio}
The generalized curvature of the average top $K$ visibility $F(\Ecal_{\Bcal})$ is given by
\begin{equation}
\alpha \approx 1-\Omega\left(\frac{1}{\xi^2} \max\left\{ \frac{1}{\sqrt{K}} , \rho e^{-\frac{1}{\rho^{2}}} \right\} \right),
\end{equation}
where, if $\rho \geq 1$, then $\alpha \approx 1-\Omega\left(\frac{1}{\xi^2}\right)$.
\end{proposition}

\xhdr{Sensor placement with submodular costs}
In sensor placement optimization~\citep{iyer2012algorithms, krause2005near, krause2008near}, the goal is typically maximizing the mutual information between the chosen 
locations $\Acal$ and the unchosen ones $\Vcal \backslash \Acal$, \ie, $F_{I}(\Acal) = I(X_{\Acal} ; X_{\Vcal \backslash \Acal})$ while simultaneously minimizing a cost function 
$c(\Acal)$ associated with the chosen locations. 
Since the mutual information is a submodular function and the costs are also often submodular, \eg, there is typically a discount when purchasing sensors in bulk, the problem 
can be reduced to maximizing a set function representable as difference between submodular functions, \ie, $F(\Acal) = F_{I}(\Acal) - \lambda c(\Acal)$, where $\lambda$ is a given parameter.
Moreover, there may be constraints on the amount of sensors in a given geographical area~\citep{powers2015sensor}, which can be represented as partition matroid constraints.

Then, we can characterize the generalized curvature of $F$ using the following Proposition, which readily follows 
from Proposition~\ref{thm:diff-submodular}:
\begin{proposition} \label{prop:sensor-placement-strong-submodularity-ratio}
Let $\alpha^{*}$ be the mininum constant for which $$\lambda [c(\Acal \cup \{v\}) - c(\Acal)] \leq \alpha^{*} [F_I(\Scal \cup \{v\}) - F_I(\Scal)].$$ 
Then, the generalized curvature of $F(\Acal)$ is $\alpha \geq \alpha^{*}$.
\end{proposition}
The above proposition assumes that the marginal gain of the cost function times $\lambda$ is always smaller than a fraction $1-\alpha^{*}$ of the marginal 
gain of the mutual information, which in turns imposes an upper bound on the given parameter $\lambda$ for which $\alpha$-submodularity holds.
\begin{figure*}[t]
\centering
\subfloat[Negative log-likelihood $-F(\Tcal)$]{\includegraphics[width=0.3\textwidth]{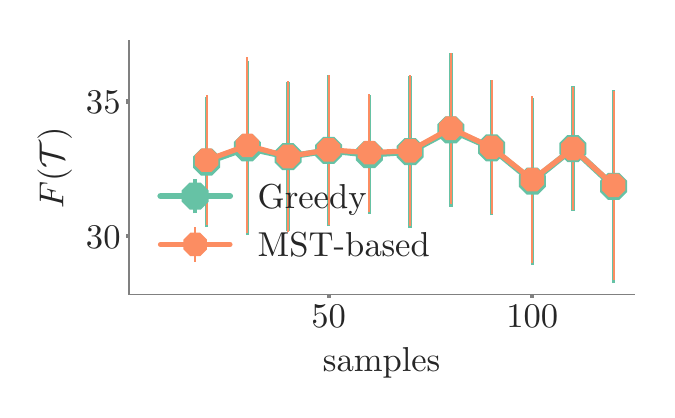}} \hspace{5mm}
\subfloat[Edge errors $|\hat{\Tcal} \backslash \Tcal|$]{\includegraphics[width=0.3\textwidth]{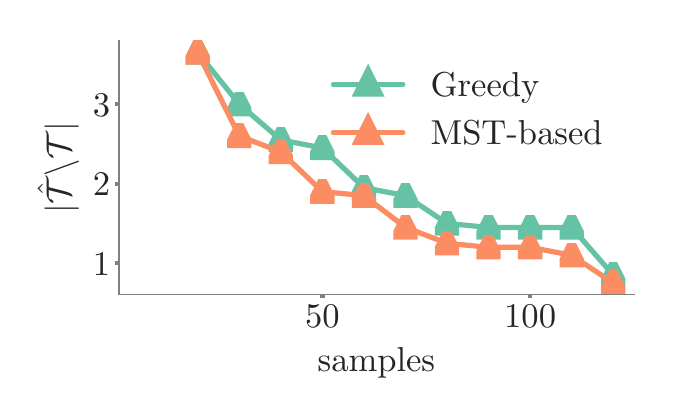}}
 \caption{Negative log-likelihood $-F(\Tcal)$ and edge errors $|\hat{\Tcal} \backslash \Tcal|$ achieved the greedy algorithm (Greedy; green) and a MST-based state of the art method (MST-based; orange) by
 Tan et al.~\citep{tan2010learning}. Each point corresponds to the average value across $20$ repetitions.} \label{fig:lineplots2}
\end{figure*}

\xhdr{More applications}
As pointed out by~\citet{iyer2012algorithms}, set functions representable as a difference between submodular functions emerged in more applications,
from feature selection and discriminatively structure graphical models and neural computation to probabilistic inference. 
In all those applications, the inequality in Eq.~\ref{eq:ds}, which must be satisfied for the functions to be $\alpha$-submodular, have a natural interpretation. 
For example, in feature selection with a submodular cost model for the features, it just imposes an upper bound on the penalty parameter that controls the 
tradeoff between the predictive power of a feature and its cost, similarly as in the case of sensor placement.

\section{Experiments}
\label{sec:Results}
In this section, our goal is to show that greedy algorithm, given by Algorithm~\ref{alg:greedy}, does achieve a competitive performance in practice. 
To this aim, we perform experiments on synthetic and real-world experiments in two of the applications introduced in Section~\ref{sec:applications},
namely, tree-structure Gaussian graphical model estimation and visibility maximization in link recommendation\footnote{\scriptsize We will release an open-source implementation of our algorithm with the final version of the paper.}.

\subsection{Tree-structured Gaussian graphical models} \label{sec:ggm-experiments}
\xhdr{Data description and experimental setup} 
%
%
%
We experiment with random trees\footnote{\scriptsize To generate a random tree, we start with an empty graph 
and add edges sequentially. At each step, we pick two of the graph's current connected components uniformly at random and connect them with an edge, choosing the end points 
uniformly at random. The process continues until there is only one connected component.} $\Tcal \in \Tcal^{n}$ with $n = 20$ vertices and edge weights $\Sigma_{ij} \sim U[0, 10]$.
We compare the performance of our estimation procedure, based on Algorithm~\ref{alg:greedy}, with the minimum spanning tree (MST) based estimation procedure 
by~\citet{tan2010learning}, which is the state of the art method, in terms of two performance metrics: 
negative log-likelihood $-F(\Tcal)$ and edge errors $| \hat{\Tcal} \backslash \Tcal |$.
Here, for our estimation procedure, at each iteration of the greedy algorithm, we solve the corresponding convex problem in Eq.~\ref{eq:mle-ggm} using 
CVXOPT~\citep{diamond2016cvxpy}.
Moreover, we run the estimation procedures using different number of samples and, for a fix number of samples, we repeat the experiment $20$ times to obtain reliable estimates
of the performance metrics.
\begin{figure*}[t]
\centering
\subfloat[$c_i = 10$]{\includegraphics[width=0.3\textwidth]{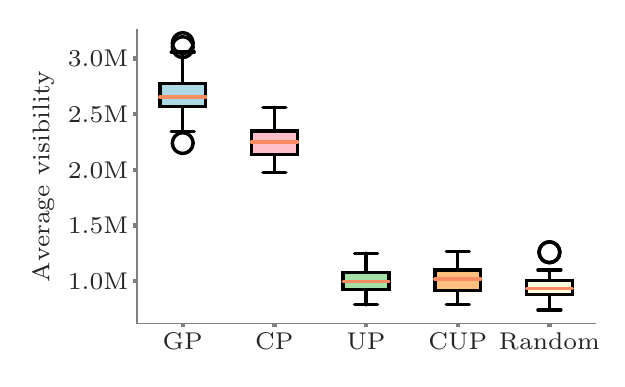}} \hspace{2mm}
\subfloat[$c_i = 20$]{\includegraphics[width=0.3\textwidth]{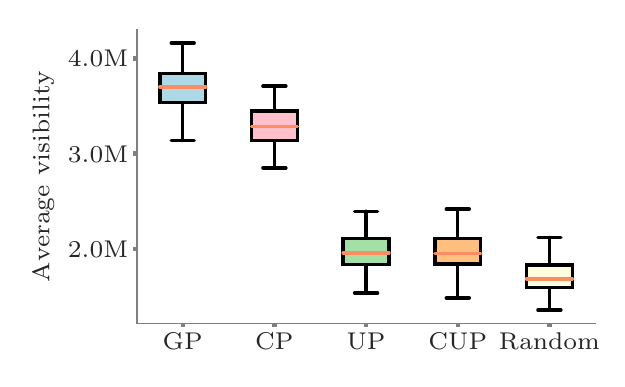}} \hspace{2mm}
\subfloat[$c_i = 40$]{\includegraphics[width=0.3\textwidth]{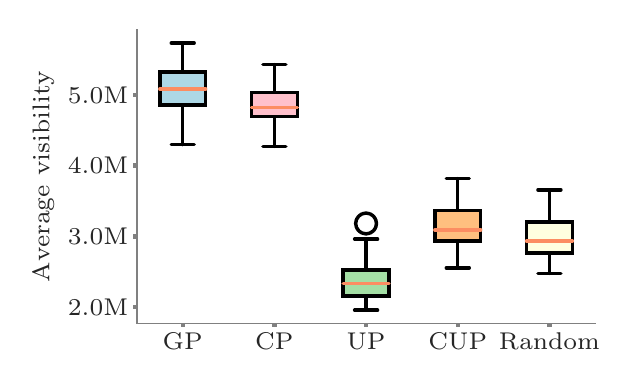}}
 \caption{Average visibility $F(\Ecal_{\Bcal})$ with $K=10$ achieved by the greedy algorithm (GP), the three heuristics (CP, UP, CUP) and the trivial baseline (Random) using Twitter data. 
 The solid horizontal line shows the median visibility and the box limits correspond to the 25\%-75\% percentiles.} \label{fig:boxplots}
\end{figure*}


\xhdr{Performance}
Figure \ref{fig:lineplots2} summarizes the results in terms of the two performance metrics, which show that both methods perform comparably, \ie, in terms of negative log-likelihood, our method 
beats the MST-based method slightly while, in terms of edge error, the MST-based method beats ours. 
%
%
%
The main benefit of using our greedy algorithm for this application is that, in contrast with MST, it provides optimality guarantees in terms of likelihood maximization.
That being said, our goal here is to demonstrate that the greedy algorithm, which is a generic algorithm, can achieve competitive performance in a structured estimation problem for 
which a specialized algorithm exists.

\subsection{Visibility optimization in link recommendation} \label{sec:visibility-real-experiments}
\xhdr{Data description and experimental setup} 
We experiment with data gathered from Twitter as reported in previous work~\citep{cha2010measuring}, which comprises user profiles, (directed) links between users, and (public) 
tweets. 
The follow link information is based on a snapshot taken at the time of data collection, in September 2009. 
Here, we focus on the tweets posted during a two month period, from July 1, 2009 to September 1, 2009, in order to be able to consider the social graph to be approximately 
static, sample a set $\Acal$ of $2000$ users uniformly at random, record all the tweets they posted. 

We compare the performance of the greedy algorithm with a trivial baseline that picks edge uniformly at random and the same three heuristics we used in the experiments
with synthetic data.
Then, for $c_i \in \{10, 20, 40\}$, we repeat the following procedure $50$ times: 
(i) we pick uniformly at random a set $\Bcal \subset \Acal$ of $80$ users as broadcasters; (ii) for each broadcaster $i$, we pick uniformly at random a set $\Hcal_i$ of $20$ 
of their followers; (iii) we record all tweets not posted by broadcasters in $\Bcal$ in the feeds of the users in $\Hcal = \cup_i \Hcal_i$; and (iv) we run the greedy algorithm, the 
heuristics, and the trivial baseline and record the sets $\Ecal_{\Bcal} = \{ (i, j) : i \in \Bcal, j \in \Hcal \}$ each provides.
Here, we run all methods using empirical estimates of the relevant quantities, \ie, $F$ using Eq.~\ref{eq:estimator} (see Appendix~\ref{app:robustness-visibility}) and $\int_{0}^{T} \gamma(t) dt$ 
using maximum likelihood, computed using the tweets posted during the first month and evaluate their performance using empirical estimates of $F$ using the tweets posted during 
the second month.

\xhdr{Solution quality}
Figure~\ref{fig:boxplots} summarizes the results by means of box plots, which show that the greedy algorithm consistently beats all heuristics and the trivial baseline. Moreover,
we did experiment with other parameters settings (\eg, $|\Hcal|$, $K$ and $c_i$) and found our method to be consistently superior to alternatives.
Appendix~\ref{app:visibility-synthetic-experiments} contains additional results using synthetic data.

\section{Conclusions}
\label{sec:conclusions}
We have 
shown that a simple variation of the standard greedy algorithm offers approximation guarantees at maximizing non-submodular 
nondecreasing set functions under a matroid constraint.
Moreover, we have identified a particular type of $\gamma$-weakly submodular functions, which we called $\xi$-submodular 
functions, for which the greedy algorithm offer a stronger approximation factor that is independent of the rank of the matroid.
In addition, we have shown that these approximation guarantees are applicable in a variety several real-world applications, from
tree-structure Gaussian graphical models and social welfare allocation to link recommendation.
%

Our work opens up several interesting avenues for future work. For example, a natural step would be to include the curvature of a set
function in our theoretical analysis, as defined in~\citet{bian2017guarantees}.
Moreover, it would be very interesting to analyze the tightness of the approximation guarantees and obtaining better upper and lower 
bounds for $\alpha$ and $\gamma$, respectively, or even unbiased estimates, for $\xi$ and $\gamma$ in the applications we considered.
Finally, it would be worth to extend our analysis to other notions of approximate submodularity~\citep{borodin2014weakly, du2008analysis, horel2016maximization,krause2008near}.

{ 
\bibliographystyle{plainnat}
\bibliography{refs}
}

\clearpage
\newpage
\onecolumn

\appendix
\section{Proof of Lemma~\ref{lem:key-weakly-submodular}} \label{app:lem-key-weakly-submodular}
Let $\Tcal$ be the optimal set of items and $\Scal$ be the set of items selected by the greedy algorithm. 
Moreover, let $\{ s_{i} \}_{i=1}^{t}$ be the items selected by the greedy algorithm in the first $t$ steps and $\{ t_{i}\}_{i=1}^{q_{t}}$ 
be the items in $\mathcal{T}$ considered by the greedy algorithm also in the first $t$ steps in order of their consideration in 
the algorithm. 
Then, we first state the following facts, that we will use throughout the proof:
\begin{itemize}[noitemsep,nolistsep,leftmargin=0.8cm]
\item[(i)] $\Scal$ and $\Tcal$ have cardinality equal to the rank of the matroid, \ie, $|\Tcal| = |\Scal| = r$. This follows from the monotonicity of
the function $F$. 
\item[(ii)] For any $t$, it readily follows that $q_t \leq t$. This follows from the proof of Theorem~\ref{thm:guarantee-xi-submodular}. 
\item[(iii)] There is a subset $\Rcal \subseteq \Tcal - \Scal_t$ with cardinality $|\Rcal| = r - t$ such that $\Scal_t \cup \Rcal \in \Ical$. This follows
from the definition of a matroid.
\item[(iv)] For all $0 \leq i \leq q_t - 1$,
\begin{equation} \label{eq:firstineq}
\rho_{s_{i+1}}(\Scal_{i}) \geq \rho_{t_{i+1}}(\Scal_{i}).
\end{equation}
This holds because, at step $i+1$, $s_{i+1}$ and $t_{t+1}$ are not considered yet and greedy selects $s_{i+1}$. Therefore, $s_{i+1}$ must have 
a higher marginal gain than $t_{i+1}$.
\item[(v)] Let $\Rcal' = \Tcal - \{t_1, \ldots, t_{q_t}\} - \Rcal$, with $|\Rcal'| = t - q_t$. Then, for all $e \in \Rcal'$ and for all $q_t \leq i \leq t-1$,
\begin{equation}
\rho_{s_{i+1}}(\Scal_i) \geq \rho_e(\Scal_i) \label{eq:secondineq}
\end{equation}
This holds because none of the items in $\Rcal'$ are considered by the greedy algorithm in the first $t$ steps. 
\end{itemize}
Now, we can use Eq.~\ref{eq:firstineq} and the fact that $\rho_{s_{i+1}}(\Scal_i) \geq \rho_{s_j}(\Scal_i)$ for $j > i$ to upper bound 
$\rho_{t_{i+1}}(\Scal_t)$ for all $0 \leq i \leq q_t - 1$:
\begin{align}
\rho_{t_{i+1}}(\Scal_{t}) & \leq \rho_{t_{i+1}}(\Scal_{t}) + \sum_{j = i}^{t-1} \rho_{s_{j+1}} (\Scal_{j}) = \rho_{\{ s_{i+1}, ... , s_{t},t_{i+1}\}}(\Scal_{i}) \nonumber \\
& \leq \frac{1}{\gamma}(\rho_{t_{i+1}}(\Scal_i) + \sum_{j=i}^{t-1} \rho_{s_{j+1}}(\Scal_{i})) \leq \frac{1}{\gamma}(t +1 - i) \rho_{s_{i+1}}(\Scal_{i}). \label{eq:maininequality1}
\end{align}
Using the same technique, if we choose an arbitrary order on the elements of $\Rcal' = \Tcal - \{t_1, \ldots, t_{q_t}\} - \Rcal = \{e_1, \ldots, e_{t-q_t}$,
we can also upper bound $\rho_{e_{i+1}}(\Scal_t)$ for all $0 \leq i \leq t - q_t - 1$:
\begin{equation}
\rho_{e_{i+1}}(\Scal_{t}) \leq \frac{1}{\gamma}(t +1 - q_{t} - i) \rho_{s_{q_{t} + i+1}}(\Scal_{q_{t} + i}) \label{eq:maininequality2}
\end{equation}
Then, it follows from Eqs.~\ref{eq:maininequality1} and~\ref{eq:maininequality2} that:
\begin{align}
\sum_{e \in \Tcal - \Rcal} \rho_{e}(\Scal_{t}) &\leq \frac{1}{\gamma}\sum_{i=0}^{t-1} (t+1-i)\rho_{s_{i+1}}(\Scal_{i}) = \frac{1}{\gamma}\sum_{i=0}^{t-1} (t+1-i) (F(\Scal_{i+1})-F(\Scal_{i})) \nonumber \\
& = \frac{1}{\gamma}\big[ F(\Scal_t) + \sum_{i=0}^{t} F(\Scal_i) \big] = \frac{1}{\gamma}\big[ 1-\frac{K_t}{OPT} + \sum_{i=0}^{t} F(\Scal_i) \big] \label{eq:secondbound}
\end{align}

In what follows, we will upper bound the sum in the above equation. To this end, we assume there exists $0 \leq \theta \leq 1$ such that it holds that $K_{t+1} \leq (1-\theta) K_{t}$. 
We will specify the value of $\theta$ later. 
Then, it follows that $K_{t} \leq (1-\theta)^{t} OPT$ and, using that $(1-\theta)^{\frac{1}{\theta}} \leq e^{-1}$,
\begin{equation}
t \leq \frac{1}{\theta} \log \left(\frac{OPT}{K_t}\right).
\end{equation}
Next, consider the geometric series $1, (1-\theta), \ldots, (1-\theta)^{\left\lfloor \frac{1}{\theta} \log \left(\frac{OPT}{K_t}\right) \right\rfloor}$ and notice
that for $j^* = \left\lfloor \frac{1}{\theta}\log(\frac{OPT}{K_{t}})\right\rfloor$, it holds that
\begin{equation*}
(1 - \theta)^{j^*} OPT \leq (1 - \theta)^{\frac{1}{\theta}\log(\frac{OPT}{K_{t}}) - 1} OPT \leq e^{-\log(\frac{OPT}{K_t})}(1-\theta)^{-1} OPT \leq \frac{K_t}{1-\theta} \leq K_{t-1}.
\end{equation*}
Then, it follows that, for each $0 \leq i \leq t-1$, there exists at least a $0 \leq j \leq \left\lfloor \frac{1}{\theta}\log(\frac{OPT}{K_{t}})\right\rfloor$ such 
that $K_{i+1} \leq (1-\theta)^j OPT \leq K_i$. 
Hence, we are ready to derive the upper bound we were looking for:
\begin{align*}
\sum_{i=0}^{t} F(S_{i}) & \leq \left( \left\lfloor \frac{1}{\theta}\log(\frac{OPT}{K_{t}}) \right\rfloor + 1 - \sum_{j=0}^{\left\lfloor \frac{1}{\theta}\log(\frac{OPT}{K_{t}}) \right\rfloor} (1-\theta)^{j} \right) OPT \\
& = \left(  \left\lfloor \frac{1}{\theta}\log(\frac{OPT}{K_{t}}) \right\rfloor - \sum_{j=1}^{\left\lfloor \frac{1}{\theta}\log(\frac{OPT}{K_{t}}) \right\rfloor} (1-\theta)^{j} \right) OPT \\
& \leq  \left( \frac{1}{\theta}\log(\frac{OPT}{K_{t}}) - (1-\theta)\frac{1-(1-\theta)^{\frac{1}{\theta}\log(\frac{OPT}{K_{t}})}}{\theta} \right) OPT \\
& \leq \left( \frac{\log(\frac{OPT}{K_{t}}) - (1-\theta)(1 - \frac{K_{t}}{OPT}) }{\theta} \right) OPT.
\end{align*}

Moreover, using the above results, we can derive an upper bound on $K_t$:
\begin{align*}
K_{t} & = OPT - F(\Scal_{t}) = F(\mathcal{T}) - F(\Scal_t) \leq F(\mathcal{T} \cup \Scal_{t}) - F(\Scal_t) \leq \rho_{\Tcal}(\Scal_{t}) \\
& \leq \frac{1}{\gamma} \sum_{e \in \Tcal} \rho_{e}(\Scal_t) = \frac{1}{\gamma} \left( \sum_{e \in \Tcal - \Rcal} \rho_{e}(\Scal_t) +  \sum_{e \in \Rcal} \rho_{e}(\Scal_t) \right) \\
& \leq \frac{1}{\gamma} \left[ \frac{1}{\gamma}\left( \frac{\log(\frac{OPT}{K_{t}}) - (1-\theta)(1 - \frac{K_{t}}{OPT}) }{\theta} + 1 - \frac{K_{t}}{OPT}\right) OPT + \sum_{e \in \Rcal} \rho_{e}(\Scal_t) \right]
\end{align*}
Therefore, if we define $e^{*} = \argmax_e \{ \rho_e(\Scal_t) \}$, then
\begin{equation*}
\rho_{e^*}(\Scal_t) \geq \frac{1}{r-t} \left[ \gamma K_t - \frac{1}{\gamma}\left( \frac{\log(\frac{OPT}{K_{t}}) - (1-\theta)(1 - \frac{K_{t}}{OPT}) }{\theta} +1 - \frac{K_t}{OPT} \right) OPT \right]
\end{equation*}
At this point, we find the value of $\theta$ such that our assumption $K_{t+1} \leq (1-\theta) K_{t}$ holds. For that, note that it is sufficient to prove that 
$\rho_{e^{*}}(\Scal_t) \geq \theta K_t$.
Hence, it is sufficient to prove that
\begin{equation*}
\frac{1}{r} \left[ \gamma K_t - \frac{1}{\gamma}\left( \frac{\log(\frac{OPT}{K_{t}}) - (1-\theta)(1 - \frac{K_{t}}{OPT}) }{\theta} +1 - \frac{K_t}{OPT} \right) OPT \right] \geq \theta K_t
\end{equation*}
To this end, we start by rewriting the above equation in terms of $\alpha = \frac{K_t}{OPT}$:
\begin{equation}
 \frac{2}{\gamma}(1-\alpha) + \frac{1}{\gamma}\frac{\log(\frac{1}{\alpha}) - 1 + \alpha}{\theta} + r \theta \alpha \leq \gamma \alpha. \label{eq:rawform2}
\end{equation}
In the above, it is easy to check that $\log(\frac{1}{\alpha}) -1 + \alpha \geq 0$. Moreover, if we fix $\alpha$, then the left hand side is minimized whenever
$\theta =  \sqrt{\frac{1}{\gamma} \frac{\log(\frac{1}{\alpha}) - 1 + \alpha}{r \alpha}}$. Then, it is sufficient to prove that
\begin{equation*}
\frac{2}{\gamma}(1-\alpha) + 2\sqrt{\frac{1}{\gamma} r\alpha \left(\log\left(\frac{1}{\alpha}\right) - 1 + \alpha\right)} \leq \gamma \alpha 
\end{equation*}
Moreover, since $\alpha \leq 1$ and $\log(\frac{1}{\alpha}) \leq \frac{1}{\alpha} - 1$, it suffices to prove that
\begin{equation*}
 \frac{1 - \alpha}{\alpha \sqrt{\gamma r}} + \frac{1}{\alpha} - 1 \leq \frac{\gamma \sqrt{\gamma}}{2\sqrt{r}},
\end{equation*}
which can be rewritten as
\begin{equation*}
\alpha \geq \frac{1}{ \frac{\gamma^2}{2(\sqrt{\gamma r}+1)}+1}.
\end{equation*}
Here, we note that definition of $\theta$ depends on $\alpha$. However, if we rewrite the inequality in Eq.~\ref{eq:rawform2} as
\begin{equation*}
\frac{2}{\gamma}\frac{1 - \alpha}{\alpha} + \frac{\log(\frac{1}{\alpha}) - 1}{\gamma \theta \alpha} + \frac{1}{\gamma \theta} + r \theta \leq \gamma,
\end{equation*}
and realize that, for $0 \leq \alpha \leq 1$, the functions $f(\alpha) = \frac{\ln(\frac{1}{\alpha}) - 1}{ \alpha}$ and $\frac{1-\alpha}{\alpha}$ are decreasing 
with respect to $\alpha$, then it is easy to see that, if the above equation holds for $0 \leq \alpha^{*} \leq 1$, it also holds for any $1 \geq \alpha = \frac{K_t}{OPT} \geq \alpha^{*}$. Hence,
if we define
\begin{equation*}
\alpha^{*} = \frac{1}{ \frac{\gamma^2}{2(\sqrt{\gamma r}+1)}+1} \quad \mbox{and} \quad \theta = \sqrt{\frac{1}{\gamma} \frac{\log(\frac{1}{\alpha^{*}}) - 1 + \alpha}{r \alpha^{*}}},
\end{equation*}
we have that $\rho_{e^{*}}(\Scal_t) \geq \theta K_t$ and thus $K_{t+1} \leq K_{t} (1-\theta)$ for any step $t$ such that $\frac{K_t}{OPT} \geq \alpha^{*}$.

\section{Proof of Proposition~\ref{prop:keyprop}} \label{app:prop:keyprop}

Let $\alpha \leq \sigma_i(\Sigma) \leq \beta$ for all $i=1, \ldots, n$. It is sufficient to prove that the curvature of the function $$f(\Sigma) = N \log | \Sigma | - \Tr \left(\Sigma \sum_{i=1}^{N} x_i x_i^{T}\right)$$ 
in any arbitrary direction $\Sigma = Rx + C$, for $\ x \in \mathbb{R}, \ R \in \mathbb{S}^n , C \in \mathbb{S}_{++}^n$ is between $\alpha^2$ and $\beta^2$.
Here, note that the possible directions in the space of positive definite matrices is equivalent to symmetric matrices, hence $R \in \mathbb{S}^n$, and we can ignore the second term
in $f(\Sigma)$ since its second derivative is zero. 
Moreover, in the remainder, we assume that the direction matrix $R$ is normalized, \ie, $|| R ||_F = 1$, where $|| \cdot ||_F$ is the Frobenius norm. 

Define $g(x) = \log | Rx + C |$. First, using that $\frac{d\log |\Sigma|}{d\Sigma} = \Sigma^{-1}$, the derivative in any specific direction is given by: 
$$\frac{dg}{dx}(x_0) = \Tr \left(\Sigma_0^{-1} R\right)$$
where $\Sigma_{0} = Rx_0 + C$. Moreover, the second derivative is readily given by
\begin{equation}
\frac{d^2g}{dx^2}(x_0) = \Tr \left(R \frac{d\Sigma^{-1}}{dx}(x_0)\right) = \Tr(R^{T} \Sigma_0 R \Sigma_0) = \Tr(R \Sigma_0 R \Sigma_0).
\end{equation}
Now, we compute an upper bound for $\Tr(R \Sigma_0 R \Sigma_0)$ as follows:
\begin{align}
\Tr(R \Sigma_0 R \Sigma_0) & \leq \sigma_{\max}(\Sigma_0) \Tr(R \Sigma_0 R)  \\
& = \sigma_{\max}(\Sigma_0) \Tr(\Sigma_0 RR) \leq \sigma_{\max}(\Sigma_0)^2 \Tr(RR) = \sigma_{\max}(\Sigma_0)^2 \Tr(R^{T}R) \\
& = \sigma_{\max}(\Sigma_0)^2 || R ||_F^2 =  \sigma_{\max}(\Sigma_0)^2 \leq \beta^2.
\end{align}
Above, we are using the inequality $\sigma_{\text{min}}(A) Tr(B) \geq Tr(AB) \leq \sigma_{\text{max}}(A) Tr(B)$ for symmetric matrix $A$ and positive definite matrix $B$, added to the fact that the matrix $R\Sigma_0R$ is positive definite because $\Sigma_0$ is positive definite, and also $RR = R^{T}R$ is positive definite. Next, we can proceed similarly to compute a lower bound and obtain
\begin{align}
\Tr(R \Sigma_0 R \Sigma_0) \geq  \sigma_{\min}(\Sigma_0)^2 \geq \alpha^2.
\end{align}
The above curvature bounds readily imply that, for given matrices $X \in S^n , Z \in D$, we have that
\begin{align*}
\alpha^2 || X ||_F^2 \leq b(X)^{T} \frac{d^2 f}{d(b(\Sigma))^2}(Z) b(X) \leq \beta^2 || X ||_F^2.
\label{curvaturebounds}
\end{align*}
which is equivalent to
\begin{align}
\alpha^2 || b(X) ||_2^2 \leq b(X)^{T} \frac{d^2 f}{d(b(\Sigma))^2}(Z) b(X) \leq \beta^2 || b(X) ||_2^2,
\end{align} 
where $b(K)$ is the vector representation of matrix $K$. Now, according to the mean value theorem, for $Y,Z \in D$, we know there is a matrix $T \in D$ on the line joining $Y$ and $Z$ such that
\begin{align}
f(Y) = f(Z) + b(Y-Z)^{T}\frac{df}{d(b(\Sigma))}(Z) + b(Y-Z)^{T} \frac{d^2f}{d(b(\Sigma))^2}(T) b(Y-Z).
\end{align}
Applying the bounds from Eq.~\ref{curvaturebounds} completes the proof.

\section{Visibility optimization in link recommendation} \label{app:visibility}

\vspace{-1mm}
\subsection{Motivation and problem definition}
\vspace{-1mm}
Users in social networks are eager to gain new followers---to grow their audience---so that, whenever they decide to share a new story, it receives a greater 
amount of views, likes and shares.
At the same time, users actually share quite a portion of their followers and, as a consequence, they are constantly competing with each other for 
attention~\citep{backstrom2011center, gomez2014quantifying}, which becomes a scarce commodity of great value~\citep{crawford2015world}.
In this context, recent empirical studies have shown that stories at the top of a user'{}s feed are more likely to be noticed and consequently liked 
or shared~\citep{hodas2012visibility, kang2015vip, lerman2014leveraging}.

The above empirical findings have motivated the recently introduced when-to-post problem~\citep{karimi2016smart, spasojevic2015post, upadhyay2018deep, zarezade2018steering, zarezade2017redqueen}, which aims to help a user, a \emph{broadcaster}, find the best times to share stories with her followers---the times when her stories would enjoy higher 
visibility and would consequently elicit greater attention from her audience.
While this line of work has shown great promise at helping broadcasters increase their visibility, it assumes the links between the broadcasters and their followers are 
given. 
However, these links are of great importance to the broadcaster'{} visibility---they define their audience---and they are (partially) influenced by link recommendation
algorithms.

The task of recommending links in social networks has a rich history in the recommender systems literature~\citep{lu2011link}. 
However, link recommendation algorithms have traditionally focused on maximizing the followers'{} utility---they recommend users to follow broadcasters whose posts
they may find interesting.
This uncompromising focus on the utility to the followers has been called into question as social media platforms are increasingly used as news 
sources\footnote{\scriptsize https://www.nytimes.com/2018/09/05/technology/lawmakers-facebook-twitter-foreign-influence-hearing.html}\footnote{\scriptsize https://www.economist.com/business/2018/09/06/how-social-media-platforms-dispense-justice}~\citep{biega2018equity, singh2018fairness}.
If we think of the broadcasters'{} utility as the visibility of their posts, the following visibility optimization problem let us balance followers'{} and 
broadcasters'{} utilities.

Given a social network $\Gcal = (\Vcal, \Ecal)$ and a set of candidate links $\Ecal'{}$ provided by a link recommendation algorithm\footnote{\scriptsize The link
recommendation algorithm may optimize for the followers'{} utilities.}, with $\Ecal'{} \cap \Ecal = \emptyset$, the goal is to find a subset of these links $\Ecal_{\Bcal} \subseteq \Ecal'{}$ 
that maximize the average visibility $F(\Ecal_{\Bcal})$ that a set of broadcasters $\Bcal$ achieve with respect to the (new) followers induced by these links, 
under constraints on the number of links per broadcaster, \ie, 
\begin{align} \label{eq:average-loss}
\underset{\Ecal_{\Bcal} \subseteq \Ecal}{\text{maximize}} & \quad F(\Ecal_{\Bcal}) \nonumber \\
\text{subject to} & \quad | \{j : (i, j) \in \Ecal_{\Bcal} \}| \leq c_i, \quad \forall i \in \Bcal
\end{align}
where we can express the constraints as $|\Bcal|$ partition matroid constraints, \ie, $\Ical = \{ \Ecal_{\Bcal} | \Ecal_{\Bcal} \subseteq \Ecal,\, |\Ecal_{\Bcal} \cap \Ecal_{i, :}'{}| \leq c_i\, \forall i \in \Bcal\}$, 
where $\Ecal_{i, :}'{}$ denotes the ground set of candidate links from broadcaster $i$,
and $c_i$ is the maximum number of links that broadcaster $i$ can \emph{afford}\footnote{\scriptsize A social media platform may charge broadcasters for 
each edge recommendation. Without loss of generality, we assume that each edge recommendation has a cost of one unit to the broadcaster and each broadcaster 
can pay for $c_i$ units.}. 
In the following sections, we denote the set of broadcasters and followers corresponding to the set of candidate links $\Ecal'$ as $\Vcal' \subseteq \Vcal$.

\vspace{-1mm}
\subsection{Definition and computation of visibility}
\vspace{-1mm}
In this section, we formally define the measure of visibility and, using previous work~\citep{karimi2016smart, zarezade2018steering, zarezade2017redqueen}, derive a relationship between 
this visibility measure and the intensity (or rates) at which broadcaster post stories and followers receive stories in their feeds. Refer to~\citet{de2019tpp} for an introduction to the 
theory of temporal point processes, which is used throughout the section.

\xhdr{Definition of visibility} 
Given a broadcaster $i$ and one of her followers $j$, let $r(t, i, j)$ be the number of stories posted by $i$ that are among the top $K$ positions of $j$'{}s feed at time $t$ 
and, for simplicity, assume each user'{}s feed ranks stories in inverse chronological order\footnote{\scriptsize At the time of writing, Twitter, Facebook and Weibo allows choosing such 
an ordering.}, as in previous work~\citep{karimi2016smart, zarezade2018steering, zarezade2017redqueen}.
Then, given an observation time window $[t_0, t_f]$ and a deterministic sequence of broadcasting events, define the deterministic top $K$ visibility of broadcaster $i$ with 
respect to follower $j$ as
\vspace{-1mm}
\begin{equation}
\Tcal(i, j) := \int_{t_0}^{t_f} r(t, i, j) dt,
\end{equation}
which is the number of stories posted by $i$'{}s that are among the top $K$ positions of $j$'{}s feed over time.
However, since the sequence of broadcasting events are generated from stochastic processes, consider instead the expected value of the top $K$ visibility 
instead, \ie, 
\vspace{-1mm}
\begin{equation} \label{eq:visibility-one-edge}
\Ucal(i, j) = \EE\left[\Tcal(i, j)\right] = \int_{t_0}^{t_f} \EE\left[ r(t, i, j) \right] dt.
\end{equation}
Then, by definition, it readily follows that
\vspace{-1mm}
\begin{equation} \label{eq:visibility-one-edge-2}
\EE\left[ r(t, i, j) \right]  = \sum_{k=1}^{K} g_{k}(t, i, j),
\end{equation}
where $g_{k}(t, i, j)$ is the probability that a story posted by broadcaster $i$ is at position $k$ of follower $j$'{}s feed at time $t$.
Finally, given a set of broadcasters $\Bcal$ and a set of links $\Ecal_{\Bcal}$, define the average top $K$ visibility (or, in short, average visibility) of the broadcasters 
with respect to the followers induced by the links as
\begin{equation} \label{eq:visibility-multiple-edges}
F(\Ecal_{\Bcal}) := \sum_{j \in \Vcal'{}} F(\Ecal_{\Bcal}, j),
\end{equation}
where, with an overload of notation, $$F(\Ecal_{\Bcal}, j) := \sum_{i \in \Bcal \,:\, (i, j) \in \Ecal_{\Bcal}} \Ucal(i, j).$$
Here, note that, by using the linearity of expectation, we can also write $F(\Ecal_{\Bcal}, j)$ in terms of 
the number of stories $r(t, \Ecal_{\Bcal}, j)$ posted by the broadcasters that are among the top $K$ positions of user $j$'{}s feed at time $t$ and 
the probability $g_{k}(t, \Ecal_{\Bcal}, j)$ that a story posted the broadcasters is at position $k$ of user $j$'{}s feed at time $t$, \ie,
\vspace{-2mm}
\begin{equation}
F(\Ecal_{\Bcal}, j) = \int_{t_0}^{t_f} \EE\left[r(t, \Ecal_{\Bcal}, j)\right] = \int_{t_0}^{t_f} \sum_{k=1}^{K} g_k(t, \Ecal_{\Bcal}, j),
\end{equation}
where we have again overloaded the notation for simplicity.

\xhdr{Computation of visibility} 
Following a similar procedure as in~\citet{zarezade2018steering, zarezade2017redqueen}, we can find a closed form expression for 
the probability $g_k(t) = g_k(t, i, j)$ that one story from a broadcaster $i$ with intensity $\mu(t)$ is at the top of a follower $j$'s feed with intensity $\gamma(t)$ at 
time $t$: 
\begin{lemma} \label{lemma:computeg}
Given a broadcaster with intensity $\mu(t)$ and one of her followers with feed intensity due to other broadcasters $\gamma(t)$, the probability
$g_{k}(t)$ that a story posted by the broadcaster is at position $k$ of the follower'{}s feed at time $t$ is given by
\vspace{-1mm}
\begin{equation} \label{eq:gk}
g_{k}(t) = \int_{0}^{t} \frac{J^{k-1}(\mu+\gamma, \tau, t)}{(k-1)!} e^{-J(\mu+\gamma, \tau, t)} \mu(\tau) d\tau,
\end{equation}
where $J(\lambda, \tau, t) = \int_{\tau}^{t} \lambda(x) dx$.
\end{lemma}

Then, if we plug Eq.~\ref{eq:gk} into Eq.~\ref{eq:visibility-one-edge-2} with $r(t) = r(t, i, j)$, we obtain
\begin{equation*}
\EE\left[ r(t) \right] = \int_{0}^{t} \left[ \sum_{k=1}^{K} \frac{J^{k-1}(\mu+\gamma, \tau, t)}{(k-1)!}\right] e^{-J(\mu+\gamma, \tau, t)} \mu(\tau) d\tau = (K-1)! \int_{0}^{t} \Gamma\left(K ,  J(\mu+\gamma, \tau, t)\right) \mu(\tau) d \tau,
\end{equation*}
where $\Gamma(K , x)$ is the incomplete gamma function. Using that $\frac{dJ(\mu+\gamma, \tau, t)}{d\tau} = -(\mu(\tau) + \gamma(\tau))$, 
we can simplify the above expression into
\begin{equation*}
\EE\left[ r(t) \right] =  - G(J(\mu+\gamma, t, t)) + G(J(\mu+\gamma, 0, t)) - \frac{1}{(K-1)!}\int_{0}^{t} \Gamma(K ,  J(\mu+\gamma, \tau, t)) \gamma(\tau) {d\tau}, 
\end{equation*}
where $G(x) =  -\sum_{i=0}^{K-1}(K-i)\frac{x^{i}}{i!} e^{-x}$ is the anti-derivative of $\Gamma(K , x)$. 
Then, using that $J(\mu+\gamma, t, t) = 0$, $G(0) = -K$ and $$J(\mu+\gamma, 0, t) = \int_{0}^{t} \mu(x)+\gamma(x) dx, $$ 
it follows that
\begin{equation}
\EE\left[ r(t) \right] = K + G\left( \int_{0}^{t} \mu(x)+\gamma(x) dx \right) - \frac{1}{(K-1)!} \int_{0}^{t} \Gamma(K ,  J(\mu+\gamma, \tau, t)) \gamma(\tau) {d\tau}. \label{eq:Ucomputation}
\end{equation}
Finally, if we plug Eq.~\ref{eq:Ucomputation} into Eq.~\ref{eq:visibility-one-edge}, we obtain an expression for the average
top $K$ visibility of broadcaster $i$ with respect to follower $j$ in terms of the intensity functions $\mu(t)$ and $\gamma(t)$
characterizing the broadcaster and the follower, respectively:
\begin{equation}
\Ucal(i, j) = K(t_{f} - t_{0}) + \int_{t_{0}}^{t_{f}} G\left( \int_{0}^{t} \mu(x)+\gamma(x) dx \right) dt - \frac{1}{(K-1)!}  \int_{t_{0}}^{t_{f}} \int_{0}^{t} \Gamma(K ,  J(\mu+\gamma, \tau, t)) \gamma(\tau) {d\tau} dt. \label{eq:Uwexact}
\end{equation}

Given a social network $\Gcal = (\Vcal, \Ecal)$, a set of candidate links $\Ecal'{}$ provided by a link recommendation algorithm, with $\Ecal'{} \cap \Ecal = \emptyset$, 
and a subset of links $\Ecal_{\Bcal} \subseteq \Ecal'{}$, we can proceed similarly as in the case of one links and show that: 
\begin{itemize}[noitemsep,nolistsep,leftmargin=0.8cm]
\item[(i)] The probability $g_{k}(t, \Ecal_{\Bcal}, j)$ that a story posted by the broadcasters user $j$ gains due to the links $\Ecal_{\Bcal}$ is at position $k$ of her feed 
at time $t$ is given by Eq.~\ref{eq:gk} with $\gamma(t) = \gamma_{j \,|\, \Ecal}(t)$, where $\gamma_{j \,|\, \Ecal}(t)$ denotes the intensity at which follower $j$ receives 
stories in her feed due to other broadcasters she follows, and $\mu(t) =  \sum_{i \in \Bcal \,:\, (i, j) \in \Ecal_{\Bcal}} \mu_i(t)$.

\item[(ii)] The average visibility $\Ucal(\Ecal_{\Bcal}, j)$ with respect to user $j$ of the broadcasters this user gains due to the links $\Ecal_{\Bcal}$ is given 
by Eq.~\ref{eq:Uwexact} with $\gamma(t) = \gamma_{j \,|\, \Ecal}(t)$ and $\mu(t) =  \sum_{i \in \Bcal \,:\, (i, j) \in \Ecal_{\Bcal}} \mu_i(t)$.
\end{itemize}
Finally, the above results allow us to write $F(\Ecal_{\Bcal})$, given by Eq.~\ref{eq:visibility-multiple-edges}, in terms of intensity functions characterizing 
the broadcasters and the feeds, as we were aiming for.

\vspace{-1mm}
\subsection{$\alpha$-submodularity of the visibility}
\vspace{-1mm}
Given a social network $\Gcal = (\Vcal, \Ecal)$, let $\Ecal_{\Bcal} \subseteq \Ecal'{}$ be a set of candidate links, with $\Ecal'{} \cap \Ecal = \emptyset$, 
$\{ \mu_i(t) \}_{i \in \Bcal}$ be the intensities of a set of broadcasters $\Bcal \subseteq \Vcal$, and $\{ \gamma_{i \,|\, \Ecal}(t) \}_{i \in \Vcal}$ 
be the users'{} feed intensities due to the broad\-cas\-ters induced by the links $\Ecal$. 
Define $\inf(\lambda) = \inf  \{ \lambda(t)  \mid t \in (t_0, t_f) \}$, $\sup(\lambda) = \sup \{ \lambda(t) \mid t \in (t_0, t_f) \}$ and assume that: 
\begin{itemize}[noitemsep,nolistsep,leftmargin=0.8cm]
\item[(i)] All intensities are bounded above, \ie, $\sup(\mu_i) \leq c_1,\, \forall i \in \Bcal$ and $\sup(\gamma_{i \,|\, \Ecal}) \leq c_1,\, \forall i \in \Vcal$ with $0 < c_1 < \infty$.
\item[(ii)] The users'{} feed intensities due to the broadcasters induced by the links $\Ecal$ are bounded below, \ie, $\inf(\gamma_{i \,|\, \Ecal}) \geq c_2,\, \forall i \in \Vcal$ 
with $0 < c_2 \leq c_1$.
\end{itemize}
Note that these assumptions are natural in most practical scenarios---the first assumption is satisfied if broadcasters post a finite number of stories 
per unit of time and the second assumption is satisfied if, at any time, there is always a nonzero probability that a user'{}s feed receives a post from 
the other broadcasters.
Moreover, under these assumptions, it readily follows that, for any set $\Ecal_{\Bcal}$, the feed intensities $\gamma_{i}(t)$
%
%
are $\xi$-bounded, \ie, $\sup(\gamma_{i}) \leq \xi \inf(\gamma_{i})$, with $\xi \leq \frac{c_{1}}{c_2}(1+|\Ecal'{}_{:, i}|)$, where $\Ecal_{i, :}'{}$ 
denotes the ground set of candidate links from broadcaster $i$.

Then, we can characterize the generalized curvature $\alpha$ of the average visibility $F(\Ecal_{\Bcal})$ using the following Theorem: 
\begin{theorem} \label{thm:first-theorem}
%
Suppose that, for any set $\Ecal_{\Bcal} \subseteq \Ecal'{}$, the feed intensities $\gamma_i(t)$ are $\xi$-bounded and, for each follower $j$,
\vspace{-1mm}
\begin{align}
\int_{0}^{t_0} \gamma_{j \,|\, \Ecal}(x) dx \geq K-1 + \zeta \sqrt{K-1}, \, \forall j \in \Vcal \label{mainconstraint}
\end{align}
for some $\zeta > 0$, which simply states that $t_0$ is large enough so that the expected number of stories posted by the broadcasters induced 
by $\Ecal$ by $t_{0}$ in each follower'{}s feed is greater than the RHS.
Then, the generalized curvature $\alpha$ of the average visibility $F(\Ecal_{\Bcal})$, defined by Eq.~\ref{eq:visibility-multiple-edges}, 
satisfies that
\begin{align}
\alpha \leq \alpha^{*} = 1 - \frac{1}{\xi \left( \frac{6.154 e^2}{ 1  -  \frac{4e^{\frac{7}{4}}}{\zeta}} \xi \sqrt{K} + 1 \right)}. \label{eq:weakratio}
\end{align}
\end{theorem}
\begin{corollary}
The generalized curvature of the average visibility $F(\Ecal_{\Bcal})$ is given by $\alpha \approx 1-\Omega\left(\frac{1}{\xi^2 \sqrt{K}}\right)$.
\end{corollary}
Moreover, if we assume that, at each time $t$, the intensity function $\mu_i(t)$ of each broadcaster is lower than a fraction of each of her 
follower'{}s feed intensities\footnote{\scriptsize This is an assumption that is likely to hold in practice since the stories posted by a single broadcaster 
are typically a small percentage of the stories her followers receive in their feeds over time.} then we can tighten upper bound $\alpha^{*}$ given by 
Theorem~\ref{thm:first-theorem} using the following Theorem and Corollary. 
\begin{theorem} \label{thm:second-theorem}
Suppose the conditions in Theorem~\ref{thm:first-theorem} hold and, in addition,
\vspace{-1mm}
\begin{equation}
\mu_{i}(t) \leq \frac{\gamma_{j \,|\, \Ecal}(t)}{\rho \sqrt{K-1}}, \quad \forall t \in [t_0, t_f], \forall i \in \Bcal, \forall j \in \Vcal'{}, \label{eq:extracondition}
\end{equation}
where $\rho > 0$. Then, the generalized curvature $\alpha$ of the average visibility $F(\Ecal_{\Bcal})$ satisfies that
\begin{equation}
\alpha \leq \alpha^{*} = 1-\frac{1}{ \alpha \left( \frac{ 2e^{\frac{7}{4}}}{(1  -  \frac{4e^{\frac{7}{4}}}{\zeta}) \min \{ \rho , 1 \} e^{-\frac{1}{\rho^2}}}\xi + 1 \right) }.  \label{eq:strongratio}
\end{equation}
\end{theorem}

Note that we can always find a constant $\rho > 0$ so that Eq.~\ref{eq:strongratio} is satisfied, however, the term $\min \{ \rho , 1 \} e^{-\frac{1}{\rho^2}}$ will decrease
drastically when $\rho \geq 1$.
Finally, by combining Theorems~\ref{thm:first-theorem} and Theorem~\ref{thm:second-theorem}, we obtain Proposition~\ref{prop:visibility-strong-submodularity-ratio} in 
the main paper.

\xhdr{Proof sketch of Theorems~\ref{thm:first-theorem} and~\ref{thm:second-theorem}}
To bound the generalized curvature of the average visibility $F(\Ecal_{\Bcal})$, defined by Eq.~\ref{eq:visibility-multiple-edges}, we have to provide a 
$\alpha$ satisfying the inequality in Eq.~\ref{eq:xi-submodularity}. 
The following Lemma 
lets us omit several sums and integrations Eq.~\ref{eq:visibility-multiple-edges} depends on while deriving a bound for the generalized curvature.
\vspace{-1mm}
\begin{lemma}	\label{lemma:aggregate}
Let $\{ F_{\sigma} \}_{\sigma \in \wp}$ be a family of set functions $F_{\sigma} : \Wcal \rightarrow \RR$ parametrized by $\sigma$, such that for each fixed $\sigma \in \wp$, 
$F_{\sigma}$ is a set function with generalized curvature $\alpha_{\sigma} \leq \alpha^{*}$. Then, the following statements hold:
\begin{itemize}[noitemsep,nolistsep,leftmargin=0.8cm]
\item[---] Suppose $\wp = (a,b) \subseteq \mathbb{R}$ and $\tilde{F}(S) = \int_{a}^{b} F_{\sigma}(S) \, d\sigma$ for every subset $S \subseteq \Vcal$ and the integral always exists. 
Then, $\tilde{F}$ has generalized curvature $\alpha \leq \alpha^{*}$. 
\item[---] Suppose $\wp$ is a discrete set and $\tilde{F}(S) = \sum_{\sigma \in \wp} F_{\sigma}(S)$. Then, $\tilde{F}$ has generalized curvature $\alpha \leq \alpha^{*}$.
\end{itemize}
For both cases introduced, if for each $\sigma$, $F_{\sigma}$ is submodular, then $F$ is submodular as well.
\vspace{-2mm}
\begin{proof}
By the definition of generalized curvature, $\forall v \in \Wcal$, $A \subseteq B$, $\forall \sigma$:
%
%
\begin{equation}
F_{\sigma}(A \cup \{ v\}) - F_{\sigma} (A) \geq (1-\alpha_{\sigma})(F_{\sigma}(B \cup \{ v\}) - F_{\sigma}(B)) \geq (1-\alpha^*)(F_{\sigma}(B \cup \{ v\}) - F_{\sigma}(B)) \label{firstprimitive}.
\end{equation}
By integrating over both sides of Eq.~\ref{firstprimitive} over $\sigma$ we get
\begin{align*}
F(A \cup \{ v\}) - F (A) \geq (1-\alpha_{\sigma})(F(B \cup \{ v\}) - F(B)) \geq 
 (1-\alpha^*)(F (B \cup \{ v\}) - F (B)),
 \end{align*}
 which proves $\alpha \leq \alpha^*$. The proof for the second part follows the same way by summing over $\sigma \in \wp$.
 \end{proof}
\end{lemma}
\vspace{-1mm}

More specifically, using the first statement of Lemma~\ref{lemma:aggregate}, in order to obtain an upper bound for the generalized curvature for $F(\Ecal_{\Bcal})$, it is 
sufficient to obtain an upper bound for each $\Ucal(\Ecal_{\Bcal}, j)$, which is a summation of three terms, as given by Eq.~\ref{eq:Uwexact}.
The first term, $K(t_{f} - t_{0})$, is a constant and does not appear in the marginal gains $\{ \rho_e \}_{e \in \Ecal_{\Bcal}}$, therefore, it does not affect the generalized curvature. 
Using the second statement of Lemma~\ref{lemma:aggregate}, it is sufficient to obtain an upper bound for the generalized curvature of the second and third term separately. 

The second term is the integration of the function $G$ in the time interval $(t_0 , t_f)$. Therefore, using the first statement of Lemma~\ref{lemma:aggregate}, it is sufficient 
to obtain an upper bound for the generalized curvature of the function $G$ for any $t \in [t_0, t_f]$. 
Here, note that the function $G$ only depends on the links pointing at follower $j$, therefore, its ground set $\Ecal_{:, j}'{}$ is the ground set of candidate links from 
the broadcasters to $j$.
Next, we will use the following Lemma 
to show that $G$ is a submodular function and thus its generalized curvature is $\alpha = 0$:
\begin{lemma} \label{lemma:concave}
Let $G(\Scal) = f(c + F(\Scal))$ be a nonnegative set function, where $f$ is a concave function over $\mathbb{R}^{+}$, $F$ is a nonnegative 
modular function, and $c$ is a nonnegative constant. Then, $G$ is submodular. 
\vspace{-1mm}
\begin{proof}
According to Lemma~\ref{lemma:concavebasic} and by using the property of modular functions, for $A \subseteq B$, it follows that
\vspace{-1mm}
\begin{align*}
G(B \cup \{ v\}) - G(B) &= f(c + F(B \cup \{ v\})) - f(c + F(B)) \\
& = f(c + F(B  \setminus A) + F(A \cup \{ v\}) ) - f(c + F(B \setminus A) + F(A)) \\
& \leq f(c + F(A  \cup \{ v\})) - f(c + F(A)) = G(A \cup \{ v\}) - G(A).
\end{align*}
\end{proof}
\end{lemma} 
\vspace{-2mm}
Moreover, the parameter of the function $G$, \ie,
\begin{equation*}
\int_{0}^{t} \sum_{i \in \Bcal \, : \, (i, j) \in \Ecal_{\Bcal}} \mu_i(x) + \gamma_{j \,|\, \Ecal}(x) dx,
\end{equation*}
for a fixed $t$, it is a sum of a modular function over the groundset $\Ecal_{:, j}'{}$ of candidate links from the broadcasters to $j$ and a 
constant. Moreover, $G$ is concave in $\mathbb{R}^{+}$ since $G'(x) = -\frac{x^{K-1}e^{-x}}{(K-1)!} \leq 0$.
Therefore, we can use Lemma~\ref{lemma:concave} to conclude that the second term is submodular over the groundset 
$\Ecal_{:, j}'{}$.
This means that an upper bound for the generalized curvature of the third term in Eq.~\ref{eq:Uwexact} will be an upper bround for the generalized curvature 
of $F(\Ecal_{\Bcal})$, on the grounds of Lemma~\ref{lemma:aggregate}. 

To upper bound the generalized curvature of the third term, using the second statement of Lemma~\ref{lemma:aggregate}, it 
is sufficient to provide an upper bound on the generalized curvature of 
\vspace{-1mm}
\begin{align}
-\int_{0}^{t} \Gamma(K ,  J(\mu+\gamma, \tau, t)) \gamma(\tau)  {d\tau} := \Pcal_j(\Ecal_{\Bcal}), \label{pdefinition}
\end{align}
where $\mu(\tau) = \sum_{i \in \Bcal \,:\, (i, j) \in \Ecal_{\Bcal}} \mu_i(\tau)$ and $\gamma(\tau) = \gamma_{j \,|\, \Ecal}(\tau)$. Here, note that the 
above function depends on the candidate links $\Ecal_{\Bcal}$ through $\mu(\tau)$.
%
Moreover,
%
%
the function $J(\mu+\gamma, \tau, t)$ is the sum of a modular function over the ground set $\Ecal_{:, j}'{}$ of candidate links from the broadcasters to $j$ and a constant.
Thus, if $-\Gamma(K , x)$ was concave for $x \in \mathbb{R}^{+}$, then we could combine the second statement of Lemma~\ref{lemma:aggregate} and 
Lemma~\ref{lemma:concave} to conclude that $\Pcal_j(\Ecal_{\Bcal})$ is submodular over $\Ecal_{:, j}'{}$. 
Unfortunately, $-\Gamma$ is convex on $(0,K-1)$ and concave on $(K-1,\infty)$. 

With this in mind, we define the time point $0 \leq \tau_{0} = \tau_{0}(\mu+\gamma, t) \leq t$ such that $J(\mu+\gamma, \tau_0, t)  = K-1$,
%
%
which will allow us to analyze the domains $(0,K-1)$ and $(K-1 , \infty)$ separately. Here, note that such $\tau_{0}$ exists because $J(\mu+\gamma, t, t) = 0$, 
\begin{align*}
J(\mu+\gamma, 0, t) &= \int_{0}^{t} \sum_{i \in \Bcal \,:\, (i, j) \in \Ecal_{\Bcal}} \mu_i(x) + \gamma_{j \,|\, \Ecal}(x) dx \geq \int_{0}^{t} \gamma_{j \,|\, \Ecal}(x) dx \geq \int_{0}^{t_0} \gamma_{j \,|\, \Ecal}(x) dx \\
&\geq K-1 \!\!\ + \zeta \sqrt{K-1} \!\! \geq K-1,
 \end{align*}
using Eq.~\ref{mainconstraint} in~Theorem~\ref{thm:first-theorem}, and the fact that $J(\mu+\gamma, \tau, t)$ is a continuous (nonincreasing) function with respect to $\tau$.

The following Lemma (proven in Appendix~\ref{app:keylemma}) introduces a key inequality to derive an upper bound on the generalized curvature of $\Pcal(\Ecal_{\Bcal})$.
 \begin{lemma} \label{keylemma}
Given a social network $\Gcal = (\Vcal, \Ecal)$. Let $\Ecal_{\Bcal}, \tilde{\Ecal}_{\Bcal}$ be two possible sets of candidate links from a set of broadcasters $\Bcal$, such 
that $\Ecal_{\Bcal} \subseteq \tilde{\Ecal}_{\Bcal}$ and $\tilde{\Ecal}_{\Bcal} \cap \Ecal = \emptyset$,
$j \in \Vcal$ a given user with feed intensity $\gamma_{j \,|\, \Ecal}(t) = \gamma(t)$ due to broadcasters induced by $\Ecal$, 
$$\mu(t) = \sum_{i \in \Bcal : (i, j) \in \Ecal_{\Bcal}} \mu_i(t) \quad \mbox{and} \quad \tilde{\mu}(t) = \sum_{i \in \Bcal : (i, j) \in \tilde{\Ecal}_{\Bcal}} \mu_i(t)$$ be 
the intensities induced by the links $\Ecal_{\Bcal}$ and $\tilde{\Ecal}_{\Bcal}$, respectively, in user $j$'{}s feed, and assume the intensities $\gamma(t)$, $\gamma(t)+\mu(t)$ and $\gamma(t)+\tilde{\mu}(t)$ to 
be $\xi$-bounded.
Consider a broadcaster $i \in \Bcal$ with intensity $\mu_i(t) = \lambda(t)$, such that $(i, j) \notin \tilde{\Ecal}_{\Bcal}$. Then, under conditions of Theorem~\ref{thm:first-theorem}, it holds that
\begin{align}
 \theta & \left[  \int_{0}^{t}  \Gamma(K ,  J(\mu+\gamma, \tau, t))  d\tau - \int_{0}^{t} \Gamma(K ,  J(\mu+\gamma+\lambda, \tau, t)) {d\tau} \right] \nonumber  \\
 \geq & \int_{0}^{t} \Gamma(K ,  J(\tilde{\mu}+\gamma, \tau, t)) d\tau - \int_{0}^{t} \Gamma(K ,  J(\tilde{\mu}+\gamma+\lambda, \tau, t)) {d\tau}, \label{eq:mainlemma1}
\end{align} 
where
\begin{equation}
\theta = \theta_{1} = \frac{6.154 e^2}{ 1  -  \frac{4e^{\frac{7}{4}}}{\zeta}} \xi \sqrt{K} + 1. \label{thetadefinition}
\end{equation}
Moreover, under extra condition of Theorem~\ref{thm:second-theorem}, Eq.~\ref{eq:mainlemma1} also holds for
\begin{align}
\theta = \theta_{2} = \frac{ 2e^{\frac{7}{4}}}{(1  -  \frac{4e^{\frac{7}{4}}}{\zeta}) \min \{ \rho , 1 \} e^{-\frac{1}{\rho^2}}}\xi + 1. \label{thetadefinition2}
\end{align}
\end{lemma}

With the above Lemma, we are now ready to derive an upper bound on the generalized curvature of $\Pcal_j(\Ecal_{\Bcal})$. Consider the same definitions and assumptions 
as in the above Lemma. Then, 
%
\begin{align*}
   \Pcal_{j}(\Ecal_{\Bcal} \cup \{ (i, j) \}) -  \Pcal_{j} (\Ecal_{\Bcal}) &=  \int_{0}^{t} \big( \Gamma(K ,  J(\mu+\gamma, \tau, t)) - \Gamma(K ,  J(\mu+\gamma+\lambda, \tau, t)) \big) \gamma(\tau)  d\tau \\
  & \geq  \inf(\gamma) \int_{0}^{t} \big( \Gamma(K ,  J(\mu+\gamma, \tau, t)) - \Gamma(K ,  J(\mu+\gamma+\lambda, \tau, t)) \big) d\tau \\
  & \stackrel{(a)}{\geq} \frac{\inf(\gamma)}{\theta} \int_{0}^{t} \big( \Gamma(K ,  J(\tilde{\mu}+\gamma, \tau, t)) - \Gamma(K ,  J(\tilde{\mu}+\gamma+\lambda, \tau, t)) \big) d\tau \\
  & \geq \frac{\inf(\gamma)}{\theta \sup(\gamma)} \int_{0}^{t} \big( \Gamma(K ,  J(\tilde{\mu}+\gamma, \tau, t)) - \Gamma(K ,  J(\tilde{\mu}+\gamma+\lambda, \tau, t)) \big) \gamma(\tau) d\tau \\
   &\stackrel{(b)}{\geq}  \frac{1}{\theta \xi} \int_{0}^{t} \big( \Gamma(K ,  J(\tilde{\mu}+\gamma, \tau, t)) - \Gamma(K ,  J(\tilde{\mu}+\gamma+\lambda, \tau, t)) \big) \gamma(\tau) d\tau  \\
   &=  (1-\alpha) \big( \Pcal_{j} (\tilde{\Ecal}_\Bcal \cup \{ (i, j) \}) -  \Pcal_{j} (\tilde{\Ecal}_\Bcal) \big), 
\end{align*}
%
%
where (a) follows from the Lemma and (b) follows from the $\xi$-boundedness of $\gamma$.
This result implies that $\alpha$, defined in Eq.~\ref{eq:weakratio} or Eq.~\ref{eq:strongratio}, is an upper bound on the generalized curvature of $\Pcal_j(\Ecal_{\Bcal})$, 
thereby, it is an upper bound on that of $F(\Ecal_{\Bcal})$ as well.

\vspace{-1mm}
\subsubsection{Proof of Lemma~\ref{keylemma}} \label{app:keylemma}
\vspace{-1mm}
We start by subtracting the left hand side (LHS) and the right hand side (RHS) of Eq.~\ref{eq:mainlemma1} and splitting the integration interval into subintervals $(0 , \tau_{0})$ and 
$(\tau_{0} , t)$ with $\tau_0 = \tau_{0}(\mu+\gamma, t)$ such that $J(\mu+\gamma, \tau_0, t)  = K-1$:
\vspace{-1mm}
\begin{align*}
LHS - RHS = \Delta(0 , \tau_{0}) + \Delta(\tau_{0} , t),
\end{align*}
where
\vspace{-1mm}
\begin{align*}
\Delta(a , b) = 
 & \theta \left[  \int_{a}^{b}  \Gamma(K ,  J(\mu+\gamma, \tau, t))  d\tau - \int_{a}^{b} \Gamma(K ,  J(\mu+\gamma+\lambda, \tau, t)) {d\tau} \right]  \\
 & - \left[ \int_{a}^{b} \Gamma(K ,  J(\tilde{\mu}+\gamma, \tau, t)) d\tau - \int_{a}^{b} \Gamma(K ,  J(\tilde{\mu}+\gamma+\lambda, \tau, t)) {d\tau} \right].
\end{align*}
For $\tau \leq \tau_{0}$, we have that $J(\mu+\gamma, \tau, t) \geq K-1$ using the fact that $J(\mu+\gamma, \tau, t)$ is nonincreasing with respect to $\tau$. Moreover, note that the intensity 
$\tilde{\mu}(t)$ is the summation of $\mu(t)$ with the intensity due to broadcasters $\tilde{\Ecal}_{\Bcal} \backslash \Ecal_{\Bcal}$. Therefore, 
\vspace{-1mm}
\begin{equation*}
J(\tilde{\mu}+\gamma, \tau, t) = \int_{\tau}^{t} \tilde{\mu}(x)+\gamma(x) dx \geq  \int_{\tau}^{t} \mu(x)+\gamma(x) dx = J(\mu+\gamma, \tau, t) \geq K-1.
\end{equation*}
In a similar way, we can conclude that
\vspace{-1mm}
\begin{align*}
J(\mu+\gamma+\lambda, \tau, t),\, J(\tilde{\mu}+\gamma+\lambda, \tau, t) \geq K-1.
\end{align*}
Then, using that 
the composite function $\Gamma(K, J(\mu+\gamma+\lambda, \tau, t))$ is convex in $J$ for $\tau \leq \tau_0$ 
and, for any convex function $f$, $x \geq x^{\prime}$, and $y \geq 0$, $f(x) - f(x + y) \leq f(x^{\prime}) - f(x^{\prime}+y)$ 
(refer to Lemma \ref{lemma:concavebasic}), it follows that
\vspace{-1mm}
 \begin{equation*}
  \Gamma(K , J(\mu+\gamma, \tau, t)) - \Gamma(K , J(\mu+\gamma+\lambda, \tau, t)) \geq \Gamma(K , J(\tilde{\mu}+\gamma, \tau, t))  - \Gamma(K , J(\tilde{\mu}+\gamma+\lambda, \tau, t)). 
 \end{equation*}
Next, we can integrate the above equation and obtain that
\vspace{-1mm}
\begin{align}
    \int_{0}^{\tau_{0}} &  \Gamma(K ,  J(\mu+\gamma, \tau, t))   d\tau - \int_{0}^{\tau_{0}} \Gamma(K ,  J(\mu+\gamma+\lambda, \tau, t)) {d\tau} \nonumber  \\
 & \geq \int_{0}^{\tau_{0}} \Gamma(K ,  J(\tilde{\mu}+\gamma, \tau, t)) d\tau - \int_{0}^{\tau_{0}} \Gamma(K ,  J(\tilde{\mu}+\gamma+\lambda, \tau, t)) {d\tau}, \nonumber
\end{align}
which implies 
\begin{align}
\Delta(0,\tau_0) \geq (\theta - 1) \big(  \int_{0}^{\tau_{0}}  \Gamma(K , J(\mu+\gamma, \tau, t))   d\tau - \int_{0}^{\tau_{0}} & \Gamma(K ,  J(\mu+\gamma+\lambda, \tau, t)) d\tau \big) \geq 0, \label{thetaminusone}
\end{align}
using that $\Gamma(K, x)$ is nonincreasing with respect to $x$ and $\theta > 1$. 
Unfortunately, $\Delta(\tau_{0} , t)$ can be negative. However, in the following, we will show that $\Delta(\tau_{0} , t) \geq -\Delta(0,\tau_{0})$. 

Let $d = J(\lambda, \tau_0, t)$. First, note that, for $\tau \leq \tau_{0}$, $J(\lambda, \tau, t) \geq J(\lambda, \tau_0, t) = d$. 
Then, starting from Eq.~\ref{thetaminusone}, we have that
\vspace{-1mm}
\begin{align*}
\Delta(0 , \tau_{0})
 &\geq 
  (\theta - 1) \int_{J(\mu+\gamma, 0, t)}^{J(\mu+\gamma, \tau_0, t)} \big[ \Gamma(K , J(\mu+\gamma, \tau, t)) - \Gamma(K ,  J(\mu+\gamma+\lambda \tau, t)) \big] \frac{d\tau}{dJ(\mu+\gamma, \tau, t)} d J(\mu+\gamma, \tau, t)  \\
  %
  &=  (\theta - 1)   \int_{K-1}^{J(\mu+\gamma, 0, t)} \frac{1}{\mu(\tau)+\gamma(\tau)} \big[ \Gamma(K , J(\mu+\gamma, \tau, t))  - \Gamma(K ,  J(\mu+\gamma+\lambda, \tau, t)) \big] d J(\mu+\gamma, \tau, t)  \\
 & \geq  \frac{\theta - 1}{\sup(\mu+\gamma)}    \int_{K-1}^{J(\mu+\gamma, 0, t)} \big[ \Gamma(K , J(\mu+\gamma, \tau, t))  - \Gamma(K ,  J(\mu+\gamma, \tau, t) + d) \big]  d J(\mu+\gamma, \tau, t)  \\
 &=  \frac{\theta - 1}{\sup(\mu+\gamma)}  \int_{K-1}^{J(\mu+\gamma, 0, t)} \big[ \Gamma(K , x)- \Gamma(K ,  x + d) \big] dx  \\
 &=  \frac{\theta - 1}{\sup(\mu+\gamma)}    \big[ \int_{K-1}^{J(\mu+\gamma, 0, t)}  \Gamma(K , x)dx -  \int_{K-1 + d}^{J(\mu+\gamma, 0, t) + d}   \!\!\!\!\!\!\!\!\!  \Gamma(K ,  x)  d x \big] \\
 &=  \frac{\theta - 1}{\sup(\mu+\gamma)}   \big[ \int_{K-1}^{K-1+d}  \Gamma(K , x)dx -  \int_{J_{s}^{t}(0)}^{J_{s}^{t}(0) + d}  \Gamma(K ,  x)  d x \big] 
\end{align*}
Next, we can use the first statement of Lemma~\ref{lemma:inequalities} (refer to Appendix~\ref{app:additional-lemmas}) to bound the second
integration term above and obtain that
\vspace{-1mm}
\begin{align}
\Delta(0 , \tau_{0}) \geq \frac{\theta - 1}{\sup(s)} \left(1 - \frac{4e^{\frac{7}{4}}}{\zeta} \right) \int_{K-1}^{K-1+d}  \Gamma(K , x)dx. \label{firstpart}
\end{align}
Second, note that, for $\tau \geq \tau_{0}$, $J(\lambda, \tau, t) \leq d$. Then, we have that
\vspace{-1mm}
\begin{eqnarray}
\begin{split}
\Delta(\tau_{0} , t)  
& \geq - \left( \int_{\tau_{0}}^{t} \Gamma(K ,  J(\tilde{\mu}+\gamma, \tau, t))d\tau  - \int_{\tau_{0}}^{t} \Gamma(K ,  J(\tilde{\mu}+\gamma+\lambda, \tau, t)) {d\tau} \right) \\
&= -  \int_{J(\tilde{\mu}+\gamma, \tau_0, t)}^{J(\tilde{\mu}+\gamma, t, t)} \big[ \Gamma(K ,  J(\tilde{\mu}+\gamma, \tau, t))  \\
& \qquad - \Gamma(K ,  J(\tilde{\mu}+\gamma+\lambda, \tau, t)) \big] \frac{d\tau}{d J(\tilde{\mu}+\gamma, \tau, t)} d J(\tilde{\mu}+\gamma, \tau, t) \\
& = - \int_{0}^{J(\tilde{\mu}+\gamma, \tau_0, t)} \frac{1}{\tilde{\mu}(\tau)+\gamma(\tau)} \big[ \Gamma(K ,  J(\tilde{\mu}+\gamma, \tau, t))  \\
& \qquad - \Gamma(K ,  J(\tilde{\mu}+\gamma+\lambda, \tau, t)) \big] d J(\tilde{\mu}+\gamma, \tau, t) \\
& \geq - \frac{1}{\inf(\tilde{\mu}+\gamma)} \int_{0}^{J(\tilde{\mu}+\gamma, \tau_0, t)} \big( \Gamma(K ,  x) - \Gamma(K ,  x + d) \big)  d x \\
& \geq - \frac{1}{\inf(\tilde{\mu}+\gamma)} \int_{0}^{\infty} \big( \Gamma(K ,  x) - \Gamma(K ,  x + d) \big)  d x \\
& = - \frac{1}{\inf(\tilde{\mu}+\gamma)}\int_{0}^{d} \Gamma(K ,  x)  d x \geq - \frac{\xi}{\sup(\tilde{\mu}+\gamma)}\int_{0}^{d} \Gamma(K ,  x)  d x, \label{secondpart}
\end{split}
\end{eqnarray}
where the last inequality follows from $\xi$-boundedness of the user'{}s feed intensity.
Before we proceed further, note that, under the extra condition of Theorem~\ref{thm:second-theorem}, we can upper bound $d$ as follows:
\begin{align*}
d = J(\lambda, \tau_0, t) \leq \!\! \int_{\tau_0}^{t} \lambda(x) dx \leq \!\! \int_{\tau_0}^{t} \frac{\gamma(x)}{\rho \sqrt{K-1}} dx 
\leq \!\! \int_{\tau_0}^{t} \frac{\mu(x)+\gamma(x)}{\rho \sqrt{K-1}}dx \\
= \frac{1}{\rho \sqrt{K-1}} J(\mu+\gamma, \tau_0, t) = \frac{\sqrt{K-1}}{\rho},
\end{align*}
and this enables us to use the third statement of Lemma~\ref{lemma:inequalities}. 
Finally, combining Eq.~\ref{firstpart} and Eq.~\ref{secondpart}, where $\theta$ is given by either Eq.~\ref{thetadefinition} or Eq.~\ref{thetadefinition2} 
together with the second or third statement of Lemma~\ref{lemma:inequalities}, depending on whether we have the extra condition of 
Theorem~\ref{thm:second-theorem}, it follows that:
\vspace{-1mm}
\begin{equation*}
 \Delta(0 , \tau_{0}) \geq \frac{\theta - 1}{\sup(s)} \left(1 - \frac{4e^{\frac{7}{4}}}{\zeta} \right) \int_{K-1}^{K-1+d}  \Gamma(K , x)dx \geq  \frac{\xi}{\sup(\tilde{s})}\int_{0}^{d} \Gamma(K ,  x)  d x \geq -\Delta(\tau_{0} , t),
\end{equation*}
which completes the proof of Lemma \ref{keylemma}.

\vspace{-1mm}
\subsubsection{Additional Technical Lemmas} \label{app:additional-lemmas}
\vspace{-1mm}
%
\begin{lemma} \label{lemma:concavebasic}
Let $f$ be a smooth function which is concave over domain $(a,b)$. Let $x,y,z$ be real numbers such that $a \leq x,y,x+z,y+z \leq b$, $x \leq y$, and $z \geq 0$. Then
$f(y+z) - f(x+z) \leq f(y) - f(x)$.
\begin{proof}
By integrating over the monotonicity relation of $f^{\prime}$.
\end{proof}
\end{lemma}

\begin{lemma}  \label{lemma:inequalities}
The following statements about the incomplete gamma function $\Gamma(K , x)$ hold:
we have the following inequalities.
\begin{itemize}[noitemsep,nolistsep,leftmargin=0.8cm]
\item[1.] Let $d \geq 0$ and $\mathcal{K} \geq K-1 + \zeta \sqrt{K-1}$, 
\begin{align}
\frac{4e^{\frac{7}{4}}}{\zeta} \int_{K-1}^{K-1+d} \Gamma(K , x) dx \geq \int_{\mathcal{K}}^{\mathcal{K}+d} \Gamma(K , x)
\end{align}
\item[2.] Let $d \geq 0$, 
\begin{align}
\frac{\int_{K}^{K-1+d} \Gamma(K,x)dx}{\int_{0}^{d}\Gamma(K,x)dx} \geq \frac{0.1625}{e^2 \sqrt{K}}
\end{align}
\item[3.] Let $0 \leq d \leq \frac{\sqrt{K-1}}{\rho}$,
\begin{align}
\frac{\int_{K}^{K-1+d} \Gamma(K,x)dx}{\int_{0}^{d}\Gamma(K,x)dx} \geq \frac{1}{2 e^{\frac{7}{4}}} \min \{ \rho,1 \} e^{-\frac{1}{\rho^2}}
\end{align}
\end{itemize}
\end{lemma}
\begin{proof}
\item[1.]  For $u \geq \sqrt{K-1}$,
\begin{align}
& \frac{\Gamma(K,\mathcal{K}+u)}{\Gamma(K,K -1 + u)} =  e^{K-1 - \mathcal{K}} \frac{\sum_{i=0}^{K-1} \frac{(K^{\prime}+u)^i}{i!}}{\sum_{i=0}^{K-1} \frac{(K - 1 + u)^i}{i!}} 
  \leq e^{K-1 - \mathcal{K}} (\frac{\mathcal{K} + u}{K -1 + u})^{K-1} \nonumber \\
 & = e^{K -1 - \mathcal{K}} \big( (1 + \frac{\mathcal{K} - K + 1}{K-1+u})^{\frac{K-1+u}{\mathcal{K} - K + 1}} \big)^{\frac{(K-1)(\mathcal{K} - K + 1)}{K-1+u}} \!\!\!\! \leq e^{K-1 - \mathcal{K}} e^{-\frac{(K-1)(K - \mathcal{K} - 1)}{K-1+u}} \nonumber \\
 & = e^{\frac{u}{K-1+u} (K -1 - \mathcal{K})} \leq e^{\frac{1}{2\sqrt{K-1}} (K-1 - \mathcal{K})} \leq e^{-\frac{\zeta}{2}} \leq  \frac{4e^{\frac{7}{4}}}{\zeta}. \label{first}
\end{align}
On the other hand, According to Lemmas \ref{thirdlowlevelinequality} and \ref{forthlowlevelinequality}, for $u \leq \sqrt{K-1}$,
\begin{equation}
\frac{\Gamma(K,\mathcal{K}+u)}{\Gamma(K,K -1 + u)} \leq  \frac{\Gamma(K,K-1 + \zeta \sqrt{K-1})}{\Gamma(K,K-1+\sqrt{K-1})}  \leq \frac{\frac{2}{e\zeta}}{\frac{1}{2e^{\frac{11}{4}}}} =  \frac{4e^{\frac{7}{4}}}{\zeta}. \label{second}
\end{equation}
Integrating over $0 \leq u \leq d$ in Eq. \eqref{first} and \eqref{second} implies the claim.
\item[2.] Define $G(d) = \frac{\int_{K}^{K-1+d} \Gamma(K,x) dx}{\int_{0}^{d} \Gamma(K,x) dx}$, $\tilde{d} = \text{argmin}_{d} G(d)$. By taking derivative from $g(d) = \frac{\Gamma(K , d)}{\Gamma(K , K-1+d)}$,
\begin{align*}
g^{\prime}(d) = \frac{e^{-d}d^{K-1}}{(K-1)!} \frac{- \Gamma(K , K-1+d) + (\frac{K-1}{d}+1)^{K-1} e^{-(K-1)}\Gamma(K,d)}{\Gamma^2(K-1+d)},
\end{align*}
we obtain that for $d \in (0,\frac{1}{e-1}(K-1))$, $g^{\prime}(d) \geq 0$. Hence,
\begin{eqnarray}
\begin{split}
\!\!\!\! G(d)  \! & \geq \! G(\tilde{d}) \! = \! \frac{F_{K}(K-1+d) - F_{K}(K-1)}{F_{K}(d)- F_{K}(0)} \! \geq \! \frac{F_{K}(K-1+\frac{K-1}{e-1}) - F_{K}(K-1)}{F_{K}(\infty)- F_{K}(0)} \\ 
& \stackrel{(a)}{ \geq} \frac{ -0.325 F_{K}(K-1) }{K} \stackrel{(b)}{ \geq} \frac{ -0.325 \frac{\sqrt{K}}{2e^2} }{K} = \frac{0.1625}{e^2 \sqrt{K}} = \Omega(\frac{1}{\sqrt{K}}), \label{secondusing}
\end{split}
\end{eqnarray}
where (a) and (b) follow from Lemmas \ref{secondlowlevelinequality} and \ref{thirdlowlevelinequality} respectively.
\item[3.] Due to the monotonicity of $\Gamma$ and using Lemma \ref{thirdlowlevelinequality}.
\begin{align*}
 G(d) \geq \frac{d\Gamma(K,K-1+d)}{d} = \Gamma(K,K-1+d) \geq \frac{1}{2e^{\frac{7}{4}}} \min \{ \rho , 1 \} e^{-\frac{1}{\rho^2}},
\end{align*}
\end{proof}

\begin{lemma} \label{firstlowlevelinequality} \vspace{-1mm}
$$-\frac{F_{K}(K-1)}{(K-1)!} \geq \frac{\sqrt{K} + 1}{2e^2}.$$
\end{lemma}
\begin{proof}
\begin{align*}
& -\frac{F_{K}(K-1)}{(K-1)!} = \sum_{k=0}^{K-1} (K -1 - k) \frac{(K-1)^k e^{-(K-1)}}{k!}  \\
& > \sum_{k=K - \sqrt{K}}^{K-1} (K -1 - k) \frac{(K-1)^k e^{-(K-1)}}{k!} \\
 & = \sum_{k=K - \sqrt{K}}^{K-1} (K -1 - k) \frac{(K-1)...(k+1)}{(K-1)^{K -1 - k}} \frac{(K-1)^{K-1} e^{-(K-1)}}{(K-1)!} \\
 & = \sum_{i = 0}^{\sqrt{K}} i \frac{(K-1)...(K - i)}{(K-1)^{i}} \frac{(K-1)^{K-1} e^{-(K-1)}}{(K-1)!} \\
 & \geq \frac{(K-1)^{K-1} e^{-(K-1)}}{(K-1)!}  \sum_{i = 0}^{\sqrt{K}} i \frac{(K-i)^{i}}{(K-1)^{i}} \\
 & =  \frac{(K-1)^{K-1} e^{-(K-1)}}{(K-1)!}  \sum_{i = 0}^{\sqrt{K}} i (1 - \frac{i-1}{K-1})^{i} \\
  & =  \frac{(K-1)^{K-1} e^{-(K-1)}}{(K-1)!}  \sum_{i = 0}^{\sqrt{K}} i \big[ (1 - \frac{(i-1)}{K-1})^{\frac{K-1}{i-1} - 1} \big]^{\frac{i(i-1)}{K - i}} \\
 &  \geq  \frac{(K-1)^{K-1} e^{-(K-1)}}{(K-1)!}  \sum_{i = 0}^{\sqrt{K}} i {e}^{-\frac{i(i-1)}{K - i}} \geq \frac{(K-1)^{K-1} e^{-(K-1)}}{(K-1)!}  \sum_{i = 1}^{\sqrt{K}} i  \\
 & \geq \frac{(K-1)^{K-1} e^{-(K-1)}}{e^2 (K-1)^{(K-1)+\frac{1}{2}} e^{-(K-1)}} \frac{\sqrt{K}(\sqrt{K} + 1)}{2} \geq \frac{\sqrt{K} + 1}{2e^2}.
\end{align*}
\end{proof}

\begin{lemma} \label{secondlowlevelinequality} \vspace{-1mm}
$$-F_{K}(K-1) \leq 1.49 -F_{K}( \frac{e}{e-1}(K-1)).$$
\end{lemma}
\begin{proof}
For RHS we can write
\begin{align*}
-& \frac{F_{K}(\frac{e}{e-1}(K-1))}{(K-1)!} = \sum_{i = 0}^{K-1} (K - i) \frac{(\frac{e}{e-1}(K-1))^i e^{-\frac{e}{e-1}(K-1)}}{i!} \\
\leq & \sum_{i = 0}^{K-1} (K - i) (K-1)^{K-1-i} \frac{(\frac{e}{e-1}(K-1))^i e^{-\frac{e}{e-1}(K-1)}}{(K-1)!} \\
\leq & \sum_{i = 0}^{K-1} (K - i) (K-1)^{K-1-i} \frac{(\frac{e}{e-1}(K-1))^i e^{-\frac{e}{e-1}(K-1)}}{\sqrt{2\pi} (K-1)^{K-1+\frac{1}{2}} e^{-(K-1)}} \\
= & \sum_{i = 0}^{K-1} (K - i) \frac{(\frac{e}{e-1})^i e^{-\frac{1}{e-1}(K-1)}}{\sqrt{2\pi (K-1)}} = \sum_{i = 0}^{K-1} \frac{K - i}{(\frac{e}{e-1})^{K -1 - i}} \frac{(\frac{e}{e-1})^{K-1}e^{-\frac{1}{e-1}(K-1)} }{\sqrt{2\pi (K-1)}} \\
= & \frac{(\frac{e}{e-1})^{K-1}e^{-\frac{1}{e-1}(K-1)} }{\sqrt{2\pi (K-1)}} \sum_{i = 0}^{K-1} (K - i)(\frac{e-1}{e})^{K -1 - i}  \\
= & \frac{(\frac{e}{e-1})^{K-1}e^{-\frac{1}{e-1}(K-1)} }{\sqrt{2\pi (K-1)}} \big[ \frac{K}{\frac{e-1}{e} - 1} (\frac{e-1}{e})^K + \frac{1 - (\frac{e-1}{e})^K }{(\frac{e-1}{e} - 1)^2} \big]  \\
\leq & \frac{(\frac{e}{e-1})^{K-1}e^{-\frac{1}{e-1}(K-1)} }{\sqrt{2\pi (K-1)}} e^2 = \frac{e^{(\ln(\frac{e}{e-1}) -\frac{1}{e-1})(K-1)} }{\sqrt{2\pi (K-1)}} e^2  \leq  \frac{e^{-0.12 (K-1)} }{\sqrt{2\pi (K-1)}} e^2.
\end{align*}
For $K \leq 20$, one can check the inequality using a computer program. For $K > 20$, according to Lemma \ref{firstlowlevelinequality},
\begin{equation}
\frac{F_{K}(K-1)}{(K-1)!} \geq \frac{\sqrt{K}+1}{2e^2} \geq \frac{\sqrt{K}}{2e^2} \geq 1.49   \frac{e^{-0.12 (K-1)} }{\sqrt{2\pi (K-1)}} e^2 \geq 1.49 \frac{F_{K}(\frac{e}{1-e} (K-1))}{(K-1)!}. \label{numericalinequalit} 
\end{equation}
\end{proof}

\begin{lemma} \label{thirdlowlevelinequality} \vspace{-1mm}
\begin{align*}
\frac{\Gamma(K , K-1+ \omega \sqrt{K-1})}{(K-1)!} \leq \frac{2}{e} w^{-1}.
\end{align*}
\end{lemma}
\begin{proof}
\begin{align*}
&\frac{\Gamma(K  , K-1+\omega \sqrt{K-1})}{(K-1)!} = \sum_{i=0}^{K-1} \frac{(K-1+\omega \sqrt{K-1})^i e^{-(K-1+\omega \sqrt{K-1})}}{i!} \\
&\leq \sum_{i=0}^{K-1} (K-1)^{K-1-i} \frac{(K-1+\omega \sqrt{K-1})^i e^{-(K-1+\omega \sqrt{K-1})}}{(K-1)!} \\
& \leq \sum_{i=0}^{K-1} (K-1)^{K-1-i} \frac{(K-1+\omega \sqrt{K-1})^i e^{-(K-1+\omega \sqrt{K-1})}}{e (K-1)^{K -1 + \frac{1}{2}} e^{-(K-1)}} \\
 & = \frac{e^{-\omega \sqrt{K-1}}}{e \sqrt{K-1}} \sum_{i=0}^{K-1} (\frac{K-1+\omega \sqrt{K-1}}{K-1})^i  \\
  & = \frac{e^{-\omega \sqrt{K-1}}}{e \sqrt{K-1}} \sum_{i=0}^{K-1} (1 + \frac{\omega}{\sqrt{K-1}})^i = \frac{e^{-\omega \sqrt{K-1} - 1}}{\sqrt{K-1}} \frac{(\frac{\omega }{\sqrt{K-1}} + 1)^K - 1}{\frac{\omega}{\sqrt{K-1}} }  \\
 & \leq \frac{e^{-\omega \sqrt{K-1} - 1}}{\sqrt{K-1}} \frac{(\frac{\omega }{\sqrt{K-1}} + 1)^{K-1}  (\frac{\omega }{\sqrt{K-1}} + 1) }{\frac{\omega}{\sqrt{K-1}} }  \\
 & \leq e^{-\omega \sqrt{K-1} - 1}   e^{ w\sqrt{K-1} }   (\frac{\omega }{\sqrt{K-1}} + 1) w^{-1} = \frac{1}{ew} (\frac{\omega }{\sqrt{K-1}} + 1)  \leq \frac{2}{e} w^{-1}. 
\end{align*}
\end{proof}

\begin{lemma} \label{forthlowlevelinequality} \vspace{-1mm}
For $t = \omega \sqrt{K-1}$, 
\begin{align*}
\frac{\Gamma(K , K-1+t)}{(K-1)!} \geq \frac{1}{2e^{\frac{7}{4}}} \min \{ \frac{1}{\omega} , 1 \} e^{-\omega^2}.
\end{align*}
\end{lemma}
\begin{proof}
Define $s = \max \{2 , 2\omega \}$. Then,
\begin{align*}
\frac{\Gamma(K,K-1+t)}{(K-1)!} &= \sum_{k=0}^{K-1} \frac{(K-1+t)^k e^{-(K-1+t)}}{k!} \geq \sum_{k=K-\frac{\sqrt{K}}{s}}^{K-1} \frac{(K-1+t)^k e^{-(K-1+t)}}{k!} \\
&= \sum_{i=0}^{\frac{\sqrt{K}}{s}} \frac{(K-1+t)^{K-1-i} e^{-(K-1+t)}}{(K-1-i)!} \\
&= \sum_{i=0}^{\frac{\sqrt{K}}{s}} \frac{(K-1)...(K-i)}{(K-1+t)^i} \frac{(K-1+t)^{K -1}e^{-(K-1+t)}}{(K-1)!} \\
&\geq \frac{(K-1+t)^{K -1}e^{-(K-1+t)}}{(K-1)!}   \sum_{i=0}^{\frac{\sqrt{K}}{s}} \big(\frac{K-i}{K-1+t} \big)^i  \\
&=  \frac{(K-1+t)^{K -1}e^{-(K-1+t)}}{(K-1)!}  \sum_{i=0}^{\frac{\sqrt{K}}{s}} \big(1 - \frac{i+t-1}{K-1+t} \big)^i  \\
&=  \frac{(K-1+t)^{K -1}e^{-(K-1+t)}}{(K-1)!}  \!\!\!\!\ \sum_{i=0}^{\frac{\sqrt{K}}{s}} \big( \big(1 - \frac{i+t-1}{K-1+t} \big)^{\frac{K-1+t}{i+t-1}+1}\big)^{\frac{i(i+t-1)}{k+2t+i-2}}  \\
&\geq  \frac{(K-1+t)^{K -1}e^{-(K-1+t)}}{(K-1)!} \sum_{i=0}^{\frac{\sqrt{K}}{s}} e^{-\frac{i(i+t-1)}{k+2t+i-2}}.  \\
\end{align*}
But according to the choice of $s$, we have
\begin{align*}
i(i+t-1) \leq \frac{\sqrt{K}}{s}(\frac{\sqrt{K}}{s} + \omega \sqrt{K} - 1) = \frac{K}{s^2} + \frac{\omega K}{s} - \frac{\sqrt{K}}{s}. 
\end{align*}
For the case $\omega \geq 1$, we have $s = 2\omega$,
$$\frac{K}{s^2} + \frac{\omega K}{s} - \frac{\sqrt{K}}{s} = \frac{K}{4\omega^2} + \frac{K}{2} - \frac{\sqrt{K}}{2\omega} \leq \frac{3}{4}K \leq \frac{3}{4}(K + 2t + i - 2).$$
In the other case where $\omega \leq 1$, we have $s = 2$,
$$\frac{K}{s^2} + \frac{\omega K}{s} - \frac{\sqrt{K}}{s} = \frac{K}{4} + \frac{\omega K}{2} - \frac{\sqrt{K}}{2} \leq \frac{3}{4}K \leq \frac{3}{4}(K + 2t + i - 2).$$
Hence, in both cases, we obtain
\begin{align*}
\frac{i(i+t-1)}{k+2t+i-2} \leq \frac{3}{4}.
\end{align*}
Therefore,
\begin{align*}
\frac{\Gamma(K,K-1+t)}{(K-1)!} &\geq  \frac{(K-1+t)^{K -1}e^{-(K-1+t)}}{(K-1)!}  \sum_{i=0}^{\frac{\sqrt{K}}{s}} e^{-\frac{3}{4}}  \\
 &\geq \frac{(K-1+t)^{K-1} e^{-(K-1+t)}}{e \sqrt{K-1} (K-1)^{K-1} e^{-(K-1)}} \frac{\sqrt{K}}{s} e^{-\frac{3}{4}} \geq \frac{1}{s} (1 + \frac{t}{K-1})^{K-1} e^{-t} e^{-\frac{7}{4}} \\
  &= \frac{1}{s} \big( (1 + \frac{t}{K-1})^{\frac{K-1}{t} + 1} \big)^{\frac{t(K-1)}{t+K-1}} e^{-t} e^{-\frac{7}{4}} \geq \frac{1}{s}e^{\frac{t(K-1)}{t+K-1}} e^{-t} e^{-\frac{7}{4}} \\
  &=\frac{e^{-\frac{7}{4}}}{s}e^{-\frac{t^2}{t+K-1}}  \geq e^{-\frac{7}{4}} \frac{1}{2}\min \{1 , \frac{1}{\omega}  \} e^{-\frac{t^2}{K-1}} =  \frac{1}{2e^{\frac{7}{4}}} \min \{1 , \frac{1}{\omega}  \} e^{-\omega^2}.
\end{align*} 
\end{proof}
\begin{figure*}[t]
	\centering
	\subfloat[$\mu_i(t)$ and $\gamma_{j \, | \, \Ecal}$]{\includegraphics[width=0.32\textwidth]{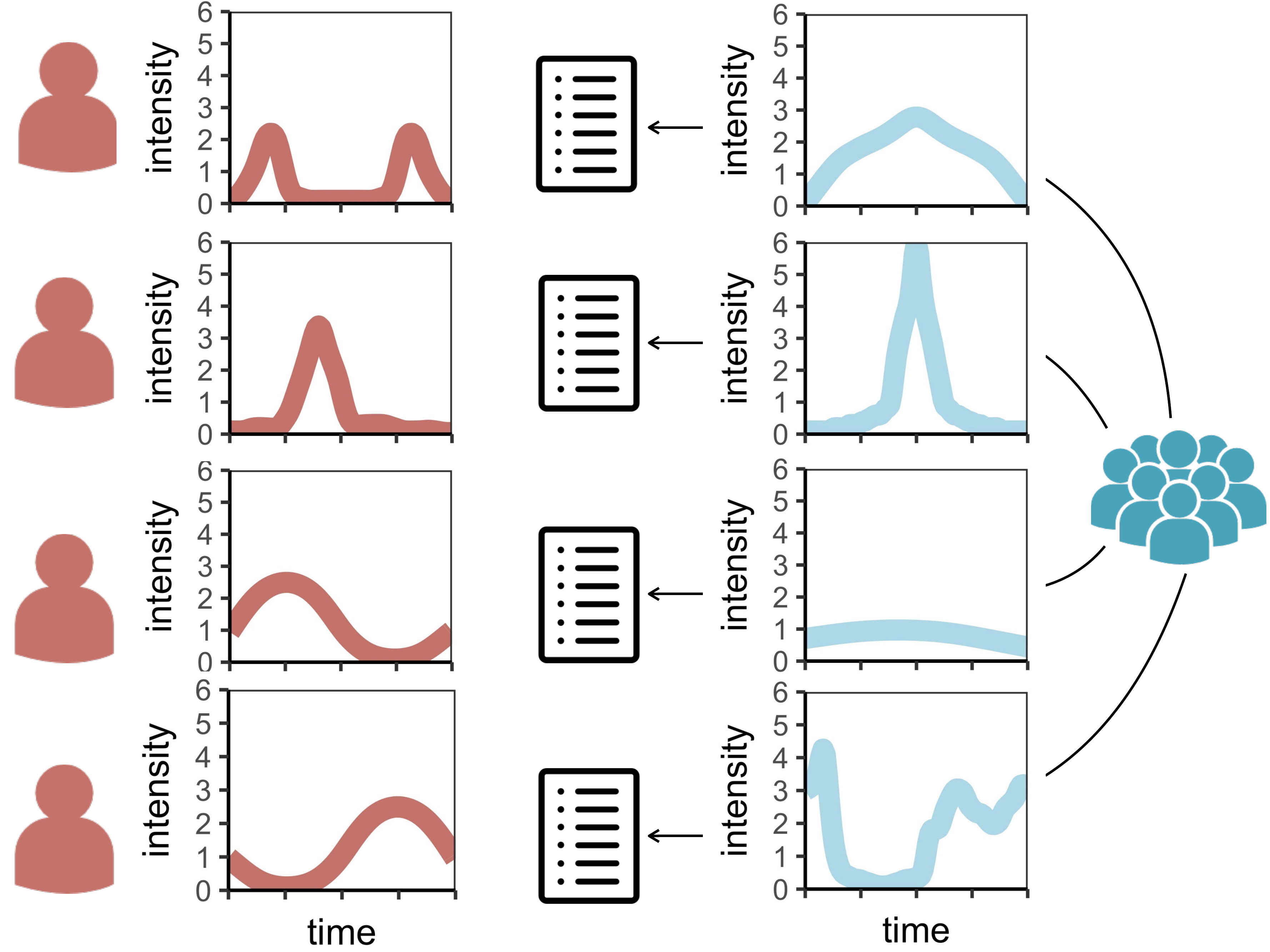}} \hspace{1mm}
	\subfloat[Greedy algorithm]{\includegraphics[width=0.32\textwidth]{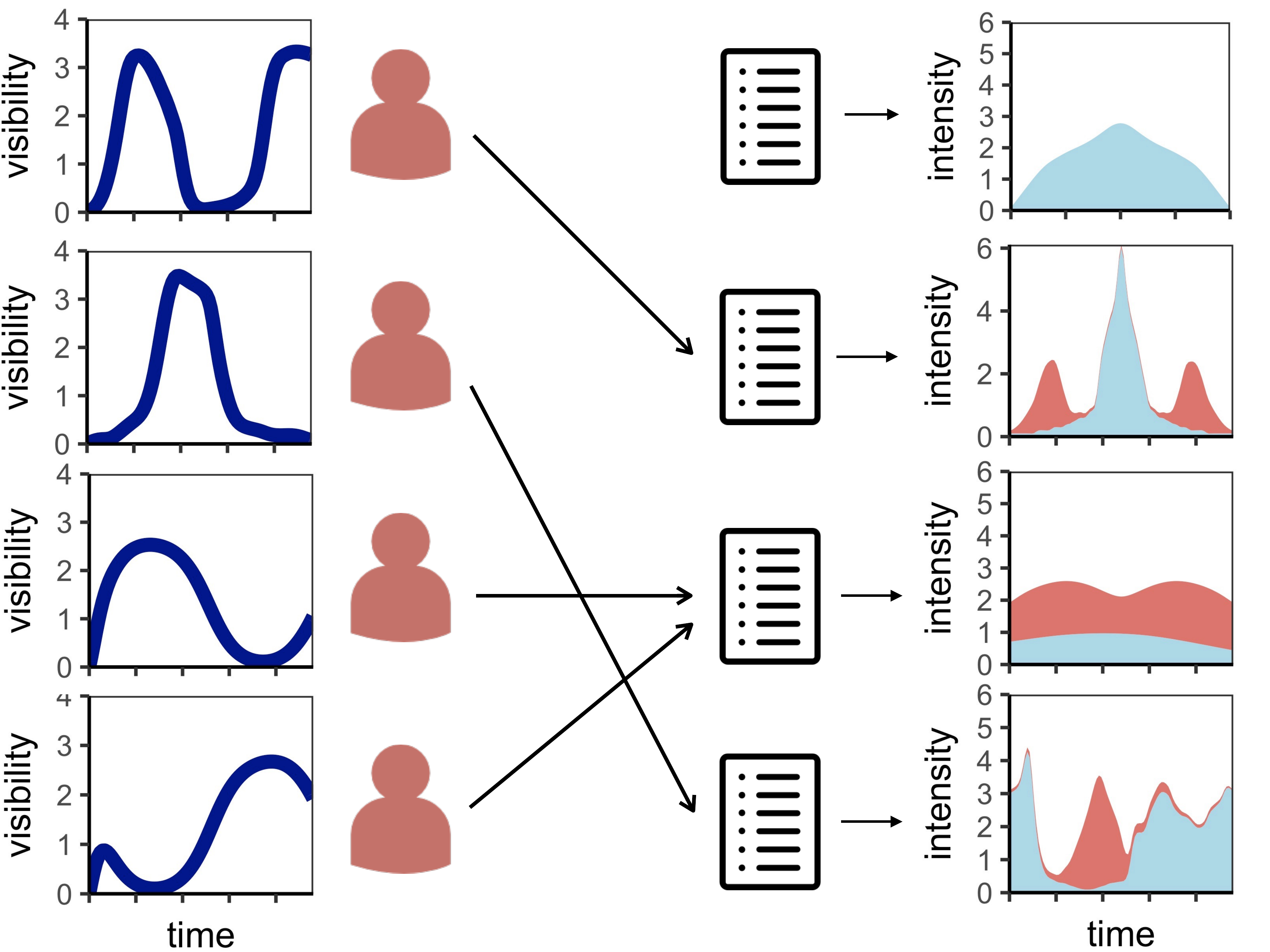}} \hspace{1mm}
	\subfloat[CP baseline]{\includegraphics[width=0.32\textwidth]{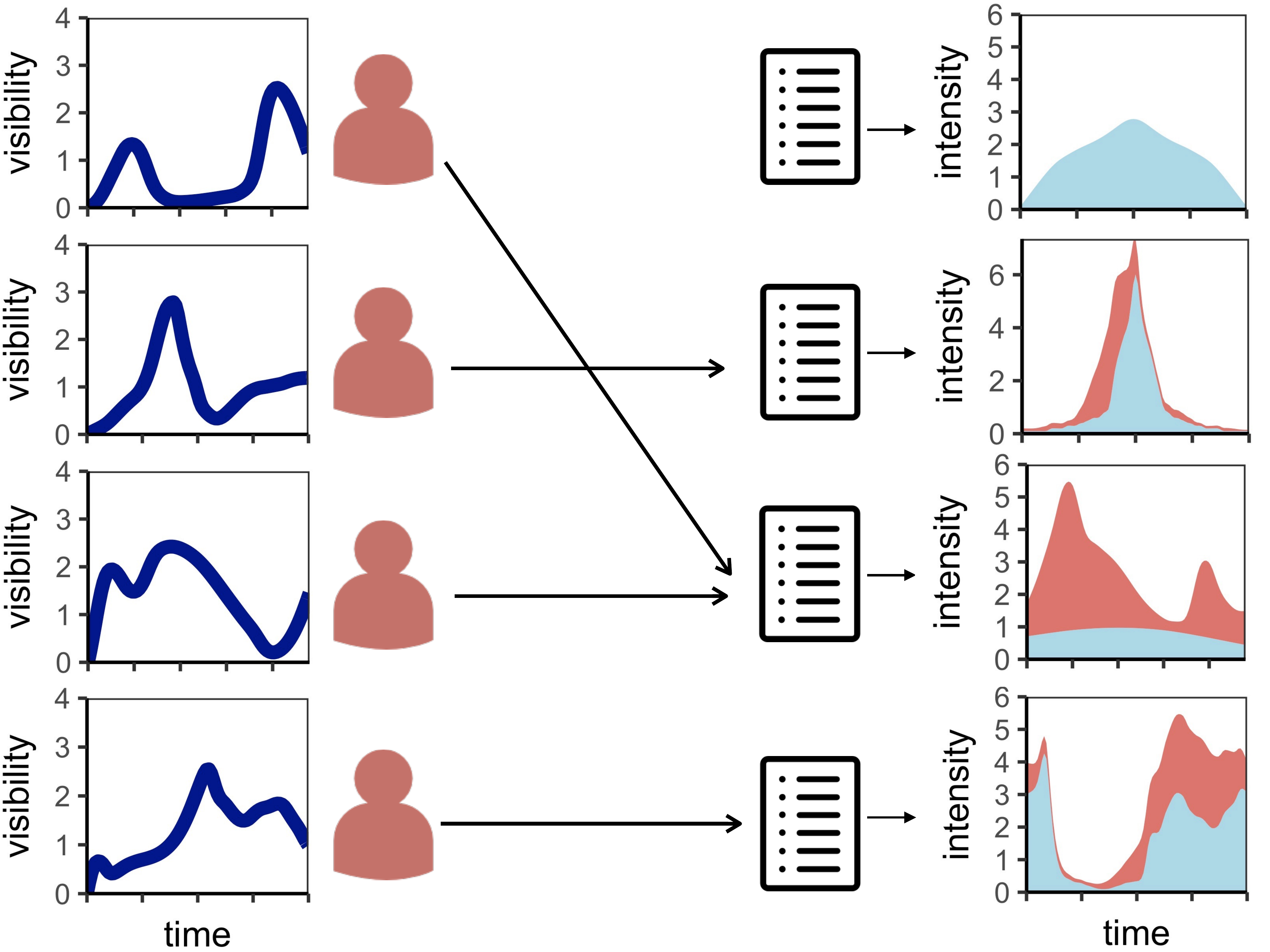}}
\caption{Toy example with four broadcasters with budget $c_i = 1$, four feeds, and $K = 10$. Panel (a) shows the broadcasters'{} intensities $\mu_i(t)$ and the feeds'{} intensities due 
to other broadcasters $\gamma_{j \,|\, \Ecal}$. Panels (b) and (c) show the solutions provided by the greedy algorithm and the CP baseline, respectively. In both panels, the left column 
shows the visibility $\Ucal(i, j)$ and the right column shows the feed'{}s intensities due to the broadcasters and the other broadcasters. }
\label{fig:toyexample}
\end{figure*}

\subsection{Robustness of the greedy algorithm}
To solve the visibility optimization problem defined in Eq.~\ref{eq:average-loss} with Algorithm~\ref{alg:greedy}, we need to compute the average 
visibility, given by Eq.~\eq{eq:visibility-multiple-edges}, which depends of a set of \emph{unknown} intensities of the broadcasters and feeds. 
In practice, one could adopt a specific functional form for these intensities and fit them using historical data, however, that could lead to poor estimates
of the visibility and, more importantly, it would be difficult to assess the impact of these empirical estimates on the approximation guarantees of the greedy
algorithm.
Instead, we can directly estimate the average visibility using historical data, derive a bound for the estimation error and assess how this estimation error 
impacts the approximation guarantees of the greedy algorithm.

\xhdr{Empirical estimation of the visibility} \label{app:robustness-visibility}
Given a directed network $\Gcal = (\Vcal, \Ecal)$, a set of candidate links $\Ecal_{\Bcal}$, with $\tilde{\Ecal}_{\Bcal} \cap \Ecal = \emptyset$, the intensities 
$\{ \mu_i(t) \}_{i \in \Bcal}$ of a set of broadcasters $\Bcal \subseteq \Vcal$, the users'{} feed intensities $\{ \gamma_{i \,|\, \Ecal}(t) \}_{i \in \Vcal}$ 
due to the broadcasters induced by the links $\Ecal$, and $n$ sequences of posts of length $\Delta = t_f - t_0$, our empirical estimate of the average visibility 
$\Ucal(\Ecal_{\Bcal}, j)$ of the broadcasters in user $j$'s feed is given by
\begin{equation}
\hat{\Ucal}(\Ecal_{\Bcal}) = \sum_{j \in \Vcal} \hat{\Ucal}(\Ecal_{\Bcal}, j) = \sum_{j \in \Vcal} \frac{\sum_{i=1}^{K} \sum_{\ell=1}^{n} \Delta^{(\ell)}_{i,j}}{n}, \label{eq:estimator}
\end{equation}
where $\Delta^{(\ell)}_{i, j}$ is the amount of time that a post from the set of broadcasters $\Bcal$ is at the $i$-th position of user $j$'{}s feed in realization $\ell$. Here, note that the empirical estimate
does not explicitly depend on the intensities of the broadcasters and feeds and, given an arbitrary set of candidate links $\Ecal_{\Bcal}$, one can always measure the visibility they
would reach without actual interventions.

We performed a formal analysis of the sample complexity of the above empirical estimate, however, for space constraints, we defer the details of this formal analysis to a 
longer version of the paper and here just state the main results. Let $0 \leq y \leq 1$, we can show that:
\begin{equation}
\PP \{ \Ucal(\Ecal_{\Bcal}, j) - y K(t_f - t_0) \leq \hat{\Ucal}(\Ecal_{\Bcal}, j) \leq \Ucal(\Ecal_{\Bcal} , j) + y K(t_f - t_0) \} \geq 1 - e^{Z_y - Q_y n}, \label{eq:majorbound} 
\end{equation}
where $Z_y$ and $Q_y$ are functions of, \eg, 
$$\beta_{j} =  \inf_{t \in (t_0, t_f)} \frac{\mu(t)}{\mu(t) + \gamma_{j \,|\, \Ecal}(t)}, \,\,
\rho_{j} =  \sup_{t \in (t_0, t_f)} \frac{\mu(t)}{\mu(t)+ \gamma_{j \,|\, \Ecal}(t)}, \,\, \mbox{and} \,\, \alpha.$$
Moreover, given the above error bound, we can characterize the approximation factor that the greedy algorithm achieves if it uses these 
empirical estimates:
\begin{theorem} \label{thm:algorithmbound}
Let the number of realizations 
$n \geq \frac{Z_y + \log \frac{  \sum_{i \in \Bcal} |\Wcal_{i,:}|  }{\delta}}{Q_y}$. 
Then, with probability at least $1 - \delta$, the greedy algorithm returns a set of links $\Ecal_{\Bcal}$ such that
\vspace{-1mm}
\begin{align*}
F(\Ecal_{\Bcal}) \geq \frac{1}{\frac{1}{\xi}+ 1} OPT - 4y K(t_f - t_0) \sum_{i \in \Bcal} c_{i},
\end{align*}
where $OPT$ is the optimal value.
%
\end{theorem}

\subsection{Experiments on synthetic data} \label{app:visibility-synthetic-experiments}

\xhdr{Experimental setup}
Unless stated otherwise, we use (periodic) piece-constant intensities $\mu_{i}(t) = \sum_{k=0}^{T-1} \mu_{i,k} \mathcal{I}(t \in [t_{k}, t_{k+1}])$ and 
$\gamma_{j \,|\, \Ecal}(t) = \sum_{j=0}^{T-1} \gamma_{j, k} \mathcal{I}(t \in [t_{k}, t_{k+1}])$ for the broadcasters and the feeds,
respectively,
where $T = 24$ days is the period, $t_{k+1} - t_k = 1$ day is the length of each piece and, each piece, we pick $\mu_{i, k}$ and $\gamma_{j, k}$ uniformly at random. 
Note that, for piece-constant intensities, we are able to compute $\Ucal(\Ecal_{\Bcal}, j)$ analytically.
We compare the performance of the greedy algorithm with the three heuristics\footnote{Initially, we also considered a trivial baseline that picks edges uniformly at random, however, its performance
was not at all competitive and decided to omit it.} from Section~\ref{sec:visibility-real-experiments}.
Here, we will run both our greedy algorithm and the baselines using the true intensity values and then report the average (theoretical) value of visibility, 
however, note that all can be run using empirical estimates of the relevant quantities, \ie, $\Ucal$ using Eq.~\ref{eq:estimator} or 
$\int_{0}^{T} \gamma(t) dt$ using maximum likelihood estimation.
\begin{figure}[t]
	\centering
	\subfloat[$F(\Ecal_{\Bcal})$ vs. $K$]{\includegraphics[width=0.32\textwidth]{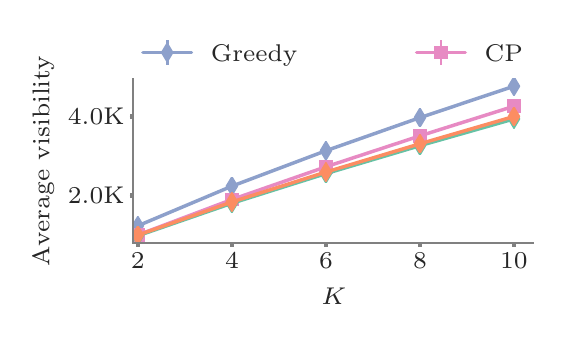}} \hspace{5mm}
        \subfloat[$F(\Ecal_{\Bcal})$ vs. $\alpha$]{\includegraphics[width=0.32\textwidth]{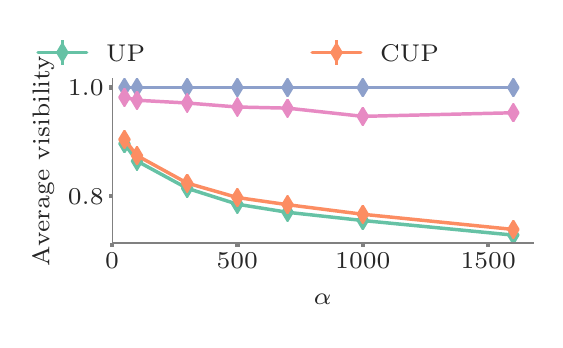}}
\caption{Average visibility achieved by the greedy algorithm (GP) and the three heuristics (CP, UP, CUP) for a setting with $60$ broadcasters and $600$ feeds. 
Panel (a) shows the average visibility, $F(\Ecal_{\Bcal})$, for different $K$ values, where we sampled $\mu_{i, k}$ and $\gamma_{j,k}$ from $U[0.01, 0.1]$ and $U[0.4, 50]$, 
respectively, and $c_i = 20$. 
Panel (b) shows the average visibility, normalized with respect to the visibility achieved by the greedy algorithm, where we sampled $\mu_{i, k}$ and $\gamma_{j,k}$ from $U[0.01, 0.1]$ and
$U[0.05, 0.05] \times \alpha$, respectively, $K = 10$ and $c_i = 50$.}
\label{fig:solution-quality-synthetic}
\end{figure}
\begin{figure}[t]
	\centering
	\subfloat[$F(\Ecal_{\hat{\Bcal}}) / F(\Ecal_{\Bcal})$ vs. $n$]{\includegraphics[width=0.32\textwidth]{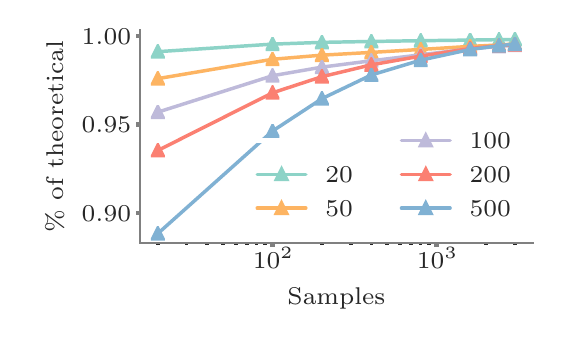} \label{fig:convergence}} \hspace{5mm}
        \subfloat[Running time vs. \# broadcasters]{\includegraphics[width=0.32\textwidth]{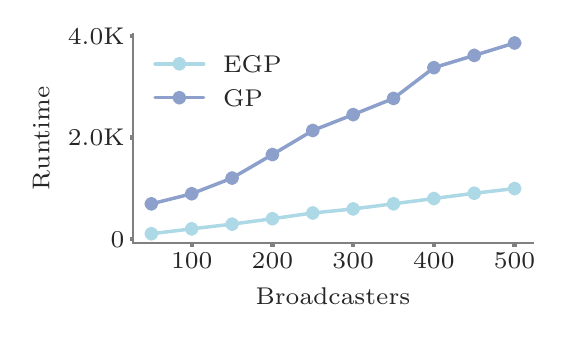} \label{fig:scalability}}
\caption{Robustness and running time of the greedy algorithm. 
Panel (a) compares the average visibility achieved by the solution $\Ecal_{\Bcal}$ ($\Ecal_{\hat{\Bcal}}$) provided by the greedy algorithm using the theoretical (empirical) visibility, where 
we sampled $\mu_{i, k}$ and $\gamma_{j,k}$ from $U[0.001, 0.1]$ and $U[0.05, 0.05 \times \alpha]$, respectively, $K = 10$ and $c_i = 30$.
Panel (b) shows the overall running time of the greedy algorithm using the theoretical (GP) and empirical (EGP) visibility, where we sampled $\mu_{i, k}$ and $\gamma_{j,k}$ from $U[0.001, 0.1]$ and $U[0.05, 0.05 \times \alpha]$, $K = 10$, and $c_i = 30$.
}
\end{figure}

\xhdr{Results}
First, we experiment with a toy example with four broadcasters, each with budget $c_i = 1$, and four feeds. Our goal here is to shed light on the way our greedy algorithm picks 
edges in comparison with one of the baselines. As illustrated in Figure~\ref{fig:toyexample}, while the greedy algorithm identifies the times when each feed'{}s intensity due to other broadcasters is low and then picks a broadcaster for each feed whose intensity is high in those times, the baseline (CP) fails to recognize such optimal matchings.

Second, we compare the performance of the greedy algorithm and all baselines in a setting with $60$ broadcasters and $600$ feeds. Figure~\ref{fig:solution-quality-synthetic} 
summarizes the results, which show that the greedy algorithm beats the baselines by large margins under different $K$ and $\alpha$ values. We did experiment with a wide 
range of parameter settings (\eg, $K$, $\alpha$, $T$ or $c_i$) and found that the greedy algorithm consistently beats the baselines.

Third, we compare the visibility values achieved by the solution $\Ecal_{\Bcal}$ the greedy algorithm provides using the theoretical visibility, given by Eq.~\ref{eq:Uwexact}, against the 
solution $\Ecal_{\hat{\Bcal}}$ it provides using the empirical visibility, given by Eq.~\ref{eq:estimator}. Figure~\ref{fig:convergence} summarizes the results, which show that, in agreement 
with Theorem~\ref{thm:algorithmbound}, the quality of the solution the greedy algorithm provides using the empirical visibility converges to the one it provides using the theoretical visibility.

Finally, we compute the running time of the greedy algorithm against the number of broadcasters. Figure~\ref{fig:scalability} summarizes the results, which show that the running time is linear
in the number of walls. 
In additional experiments, we also found that the running time is linear in the number of walls, superlinear with respect to the number of pieces $T$ and it is independent on the
budget per broadcaster, however, for space constraints, we do not include the corresponding plots.

\end{document}